\documentclass[journal]{IEEEtran}
\usepackage{color}
\usepackage{verbatim}
\usepackage{amsfonts}
\usepackage{amssymb}
\usepackage{stfloats}
\usepackage{cite}
\usepackage{graphicx}
\usepackage{psfrag}
\usepackage{subfigure}
\usepackage{amsmath}
\usepackage{array}
\usepackage{epstopdf}
\usepackage{authblk}
\usepackage{graphicx} 
\usepackage{amsthm} 
\usepackage{lipsum}
\usepackage{verbatim} 
\usepackage{authblk}
\usepackage{mathtools}
\usepackage{cuted}
\usepackage[lined,boxed,ruled]{algorithm2e}
\usepackage{algpseudocode}
\usepackage{framed} 
\usepackage{subfigure}
\usepackage{soul}
\usepackage{bm}
\usepackage{setspace}
\usepackage{url}
\usepackage{booktabs}

\newtheorem{remark}{\bf Remark}

\def\qed{$\Box$}

\def\phi{\varphi}

\def\({\left(}
\def\){\right)}

\def\bff{{\mathbf{f}}}
\def\bg{{\mathbf{g}}}

\def\br{{\mathbf{r}}}
\def\bs{{\mathbf{s}}}
\def\bt{{\mathbf{t}}}

\def\b0{{\mathbf{0}}}

\def\bA{{\mathbf{A}}}

\def\bS{{\mathbf{S}}}

\usepackage{tabularx}

\newcounter{protocol}

\newtheorem{Lemma}{Lemma}

\newtheorem{Proposition}{Proposition}

\IEEEaftertitletext{\vspace{-2\baselineskip}}

\title{Over-the-Air Fusion of Sparse Spatial Features for Integrated Sensing and Edge AI over Broadband Channels}

\author{Zhiyan Liu, Qiao Lan, and Kaibin Huang
\thanks{
Manuscript received 26 April 2024; revised 10 September 2024; accepted 24 December 2024. The work described in this paper was supported in part by the Research Grants Council of the Hong Kong Special Administrative Region, China under a fellowship award (HKU RFS2122-7S04), the Areas of Excellence scheme grant (AoE/E-601/22-R), Collaborative Research Fund (C1009-22G), and the Grant 17212423. Part of the described research work is conducted in the JC STEM Lab of Robotics for Soft Materials funded by The Hong Kong Jockey Club Charities Trust. An earlier version of this paper was presented in part at the IEEE International Conference on Communications Workshops (ICC Workshops), Denver, CO, USA, June 9--13, 2024 [DOI: 10.1109/ICCWorkshops59551.2024.10615471].  The associate editor coordinating the review of this article and approving it for publication was Y. Zhang.
\emph{(Corresponding author: Kaibin Huang.)} 

Z. Liu, Q. Lan and K. Huang are with Department of Electrical and Electronic Engineering at The University of Hong Kong, Hong Kong (e-mail: zyliu@eee.hku.hk; qlan@eee.hku.hk; huangkb@eee.hku.hk).  }}

\makeatletter
\newcommand{\removelatexerror}{\let\@latex@error\@gobble}
\makeatother

\IEEEoverridecommandlockouts

\begin{document}
\maketitle
\begin{abstract}
The \emph{sixth-generation} (6G) mobile networks feature two new usage scenarios -- distributed sensing and edge \emph{artificial intelligence} (AI). 
Their natural integration, termed \emph{integrated sensing and edge AI} (ISEA), promises to create a platform that enables intelligent environment perception for wide-ranging applications. 
A basic operation in ISEA is for a fusion center to acquire and fuse features of spatial sensing data distributed at many edge 
devices (known as agents), which is confronted by a communication bottleneck due to multiple access over hostile wireless channels. 
To address this issue, we propose a novel framework, called \emph{Spatial Over-the-Air Fusion} (Spatial AirFusion), which exploits radio waveform superposition to aggregate spatially sparse features over the air and thereby enables simultaneous access. 
The framework supports simultaneous aggregation over multiple voxels, which partition the 3D sensing region, and across multiple subcarriers. 
It exploits both spatial feature sparsity with channel diversity to pair voxel-level aggregation tasks and subcarriers to maximize the minimum receive signal-to-noise ratio among voxels. 
Optimally solving the resultant mixed-integer problem of \emph{Voxel-Carrier Pairing and Power Allocation} (VoCa-PPA) is a focus of this work. 
The proposed approach hinges on derivations of optimal power allocation as a closed-form function of voxel-carrier pairing and a useful property of VoCa-PPA that allows dramatic solution space reduction. 
Both a low-complexity greedy algorithm and an optimal tree-search algorithm are then designed for VoCa-PPA. 
The latter is accelerated with a customised compact search tree, node pruning and agent ordering.
Extensive simulations using real datasets demonstrate that Spatial AirFusion significantly reduces computation errors and improves sensing accuracy compared with conventional over-the-air computation without awareness of spatial sparsity. 

\end{abstract}


\vspace{-2mm}
\section{Introduction}
\label{sec: introduction}
The \emph{sixth-generation} (6G) mobile network warrants two essential capabilities, sensing and \emph{artificial intelligence} (AI)~\cite{ITUR2023}.  
The first capability involves the integration of diversified sensing modalities such as camera, mmWave, and \emph{Light Detection and Ranging} (LiDAR) sensors to collect information from sensory data. 
The second capability is envisioned to support AI model deployments in 6G edge networks, enabling the delivery of intelligent services. 
Integrating these two essentials for advanced 6G applications, ranging from high-precision perception to human-machine symbiosis, leads to an emerging paradigm called \emph{Integrated Sensing and Edge AI} (ISEA)~\cite{Chen2023arxiv}.
In such a system, an edge device equipped with sensors, termed an \emph{agent}, in a distributed sensing system acquires sensory data from its surroundings and sends features extracted using its local perception model to the edge server (i.e., fusion center) for aggregation and then inference {to support intelligent decisions and real-time actions} for a downstream AI application\cite{Huang2023TWC,Dingzhu24TWC}. 
However, ISEA faces a communication bottleneck due to the aggregation of high-dimensional sensing features over resource-constrained wireless channels~\cite{Wu2023Network,Shi2022JSAC}. 
One promising solution for overcoming the bottleneck is called \emph{Over-the-Air Computation} (AirComp), which exploits waveform superposition in simultaneous access to realize over-the-air data aggregation\cite{GX2021WCM,GX2020TWC,Deniz2020TWC,CMZ2021JSAC}. 
Based on AirComp, we develop a novel framework, termed \emph{Spatial Over-the-Air Fusion} (Spatial AirFusion), for communication-efficient multi-sensor fusion in environment perception over a broadband channel. 
Its distinctive feature, { differentiating it from conventional AirComp,} is to exploit spatial feature sparsity and channel frequency selectivity to intelligently map voxels, which divide the sensing region, to subcarriers for performing voxel-level AirComp tasks.
Thereby, the sensing performance is improved while computation complexity reduced. 

Precise environment perception underpins a set of killer application scenarios of 6G, e.g., autonomous driving and collaborative robots. 
State-of-the-art perception models~\cite{Zhou2018CVPR} leverage LiDAR, mmWave, and camera data to generate spatial feature vectors associated with certain locations in the physical world, as opposed to location-agnostic features in conventional classification and object detection. 
This type of feature is known as \emph{voxel features}, where one voxel represents a spatial region in an evenly spaced 3D grid of the sensing range\cite{Rukhovich2022WACV,Xie2022arxiv}. 
To support low-latency and large-scale environment perception in 6G networks requires task-oriented air-interface design targeting ISEA. 
As a specific use case of edge inference, ISEA can be implemented on the well-known split inference architecture\cite{Lan2023TWC,Shao2023TWC,Zhang2020CM,Zhou2020IoTJ}. 
In this architecture, a global inference model is split into a device and a server sub-model with the former used for local feature extraction and the latter for remote inference\cite{Zhang2020CM}. 
It can be generalized to distributed split inference (for which ISEA is a special case) by deploying models at multiple devices for local feature extraction (or inference) and performing local-feature (or label) aggregation at the server to attain a high inference accuracy\cite{Shao2023TWC,Deniz2022ISIT,Wen2023TWC}. 
{ For communication-efficient feature aggregation, an AirComp-based general framework is proposed in \cite{Huang2023TWC} to realize different feature-aggregation functions, which include maximization,  in an ISEA system based on an end-to-end sensing performance metric. As a simultaneous-access technology,  AirComp promises to solve the scalability issues in ISEA, enabling low-latency device cooperation, which motivates us to further investigate AirComp for ISEA-based cooperative perception. }

AirComp in its own right is a fast-growing area\cite{GX2021WCM}.
The principle of AirComp is to exploit the superposition of signals simultaneously transmitted by multiple agents such that the desired aggregation functions, e.g., averaging, multiplication, and maximization, can be realized over the air\cite{Gastpar2011TIT,Katabe2016arxiv}. 
To materialize accurate functional computation via AirComp requires coping with channel fading and noise.
For this purpose, a line of techniques has been designed to minimize AirComp errors including optimal power control\cite{Xiaowen2020TWC}, \emph{multiple-input-multiple-output} (MIMO) beamforming\cite{GX2019IOTJ} and interference management\cite{Xiaowen2021TWC}. 
Broadband transmission is prevalent in modern high-rate mobile systems, which is assumed in the current system model. 
This motivates researchers to study broadband AirComp by addressing issues such as power allocation among subcarriers\cite{Wen2023TWC,Bennis2021Globecom}, subcarrier truncation to avoid deep fading\cite{GX2020TWC} and exploitation of channel frequency diversity~\cite{Liu2021WCL}. 
The co-existing information-transfer users and AirComp devices participating in federated learning are also studied where the rate-maximizing subcarrier allocation for the former is designed subject to a guarantee on the learning performance of the latter~\cite{Zhang2023TWC}. 
{ Going beyond computation of generic aggregation functions, AirComp can be applied and tailored for specific AI computation tasks. This idea of task-oriented AirComp design originated in AirComp applications in \emph{federated learning} (FL), which created an area called \emph{over-the-air FL} (AirFL)\cite{GX2020TWC,Deniz2020TWC,CMZ2021JSAC}. }
In this paradigm, AirComp realizes over-the-air aggregation of local gradients or models uploaded by  devices, from which the result is used to update a global model at an edge server\cite{GX2020TWC,Deniz2020TWC}.
While traditional AirComp techniques aim at computation error minimization, the design objective of AirFL techniques is to accelerate learning and account for the specific characteristics of transmitted data (i.e., local gradients/models). 
This results in a rich set of task-oriented wireless techniques such as power control based on gradient statistics\cite{Tao2021TWC}, data- and channel-aware sensor scheduling\cite{ZJun2024TWC}, adaptive precoding\cite{Eldar2021TSP}, etc. 

Existing studies on AirComp as discussed above all assume single-stream data sources without considering data spatial distributions. 
Nevertheless, spatial feature variation is a key characteristic of environment perception as reflected in two aspects.  
{First, features are sparsely distributed in the voxel dimension. At the outputs of prevalent sensing models (e.g., VoxelNet\cite{Zhou2018CVPR} and PointPillars\cite{pointpillars}), features for a given voxel are non-zero only if the voxel contains detectable objects in the physical world (e.g., vehicles and pedestrians). Consequently,  only a small portion of all voxels are nonzero due to finite sensing ranges, view occlusion, and sparse scattering of  objects in  space. For example, both\cite{Zhou2018CVPR} and\cite{pointpillars} report over 90\% empty voxels.
Second, }spatial feature distributions as observed by different agents are heterogeneous because of their non-identical fields of perception and view angles.  
Another aspect of heterogeneity is { multiuser and frequency diversities of wireless channels. }
One key effect of spatial feature variation is the spatial variation of AirComp error as elaborated in the sequel. 
Let the task of spatial feature aggregation be divided into voxel-level sub-tasks. 
Due to the sparsity and heterogeneity of spatial feature distributions, the subset of agents participating in aggregation varies from voxel to voxel. {This results in different AirComp errors for different voxels as the errors }depend on the numbers of participating agents (see, e.g.,\cite{Vucetic2020TWC}) and qualities of the associated channels. 
The errors can be manipulated using a mechanism called \emph{Voxel-Carrier (VoCa) Pairing} that maps voxels to subcarriers for executing their sub-tasks. 
Via this mechanism, a large number of degrees-of-freedom due to numerous voxels and subcarriers can be exploited to improve the performance of Spatial AirFusion. 
Furthermore, VoCa Pairing can be integrated with power allocation over subcarriers to obtain additional performance gain, giving rise to the problem of optimal \emph{VoCa Pairing and Power Allocation} (VoCa-PPA).

Let the performance of Spatial AirFusion be measured using the metric of the minimum receive SNR among all voxels, which serves as an indicator of the largest AirComp error. 
Given the objective of maximizing the metric, a subcarrier under favourable channel conditions should be ideally paired with a voxel with many participating agents. 
However, given the heterogeneity in multiple voxels and sub-channels of multiple agents, the optimal VoCa-PPA problem becomes a sophisticated mixed integer program. 
In this work, we present the framework that consists a set of algorithms for efficiently solving the problem via exploiting the unique features of Spatial AirFusion. 
The key contributions are summarized as follows. 
\begin{itemize}
    \item \textbf{AirFusion Protocol.} A communication protocol is presented to realize spatial AirFusion in a multi-agent system, comprising the following three phases. 
    First, each agent sends binary sparsity indicators of all voxels in the sensing region to the server. 
    In the second phase of radio resource allocation, the server performs VoCa-PPA using one of the proposed algorithms based on the sparsity indicators and broadband channel states. 
    Last, in the over-the-air fusion phase, the agents’ feature vectors on voxels are transmitted simultaneously and aggregated over the air using the assigned subcarriers and power. 
    \item \textbf{Greedy VoCa-PPA Algorithm. } A low-complexity algorithm is designed to compute a sub-optimal solution for the VoCa-PPA problem by sequentially solving the problems of optimal power allocation and VoCa Pairing. 
    First, given VoCa Pairing, the optimal allocated power for subcarriers is derived in closed-form. As revealed by the result, the minimum receive SNR  depends solely on a bottleneck agent charaterized by poorest associated channels. 
    Second, given the derived power allocation, the VoCa-PPA problem is reduced to the problem of optimal VoCa Pairing, which is combinatorial and NP-hard\cite{ehrgott2006discussion}. 
    It is solved using a low-complexity greedy algorithm that iteratively matches each voxel with the best-matched subcarrier under the criterion of minimizing the maximum channel-inversion power over all participating agents. 
    In this regard, voxels with relatively high feature densities tend to involve more agents participating in AirComp, which degrades receive SNRs. For this reason, they are given higher priorities so as to be matched to better sub-channels. 
    \item \textbf{Optimal VoCa-PPA Algorithms. } Leveraging the optimal power allocation derived previously simplifies the optimal VoCa-PPA problem to optimal VoCa Pairing without sacrificing the solution's optimality.  
    Despite a simpler form, the latter is a max-linear assignment problem that does not admit polynomial-time solutions. 
    To address the issue, a solution approach is designed to significantly reduce the computation complexity. 
    The approach is comprised of two designs - a compact search tree and a \emph{depth-first search} (DFS) algorithm that are both customised for VoCa Pairing. 
    Underpinning these algorithms is a useful property of the problem that two voxels with identical sparsity indicators are equivalent from the perspective of minimizing the objective. 
    The property is exploited to convert the VoCa Pairing problem from the original \emph{one-to-one mapping} to \emph{subset-to-subset mapping}. 
    As a result, orders-of-magnitude reduction in computation complexity is achievable. 
    The complexity of tree search is further reduced using two proposed schemes.  
    The first is intelligent early stopping and node pruning based on criteria developed by comparing the current best global objective and local objectives in each step. 
    The other is agent ordering in DFS based on a designed priority indicator combining each agent's channel states and sparsity pattern.
    \item \textbf{Experiments.} The performance of Spatial AirFusion is evaluated by extensive experiments using both synthetic and real datasets (i.e., OPV2V\cite{xu2022opencood}). { The benchmarking schemes include 1) \emph{naive AirComp} which schedules all sensors for all voxels without sparsity awareness; 2) \emph{AirFusion-Vanilla} which adopts the sparsity-aware framework but randomly pairs voxels with subcarriers; 3) \emph{digital air interface} where devices transmit orthogonally. }
    The proposed framework is demonstrated to outperform naive AirComp by a large margin, e.g., 10 dB gain in AirComp error suppression and significantly improved end-to-end inference accuracy. Compared with digital air interface, AirFusion achieves $5.74$ times reduction in communication latency with the same inference accuracy. 
\end{itemize}

\section{System Models}
\label{sec: models-and-metrics}

We consider an ISEA system targeting environment perception, where $K$ agents are distributed in the space and cooperate to complete a sensing task as coordinated by a \emph{fusion center} (FC). The system is illustrated in Fig.~\ref{fig_system} for the context of autonomous-driving perception where agents are helper vehicles and the fusion center is an ego vehicle. For each perception instance, each agent acquires a view (e.g., a LiDAR frame) of the surrounding environment via its sensor and extracts its local features. The fusion center then employs an AirFusion technique as proposed in subsequent sections to wirelessly aggregate local features and perform inference for the global perception results. Relevant models and the performance metric are described in the following subsections. 

\begin{figure*}[t]
\centering
\subfigure[]{\includegraphics[width=0.61\textwidth]{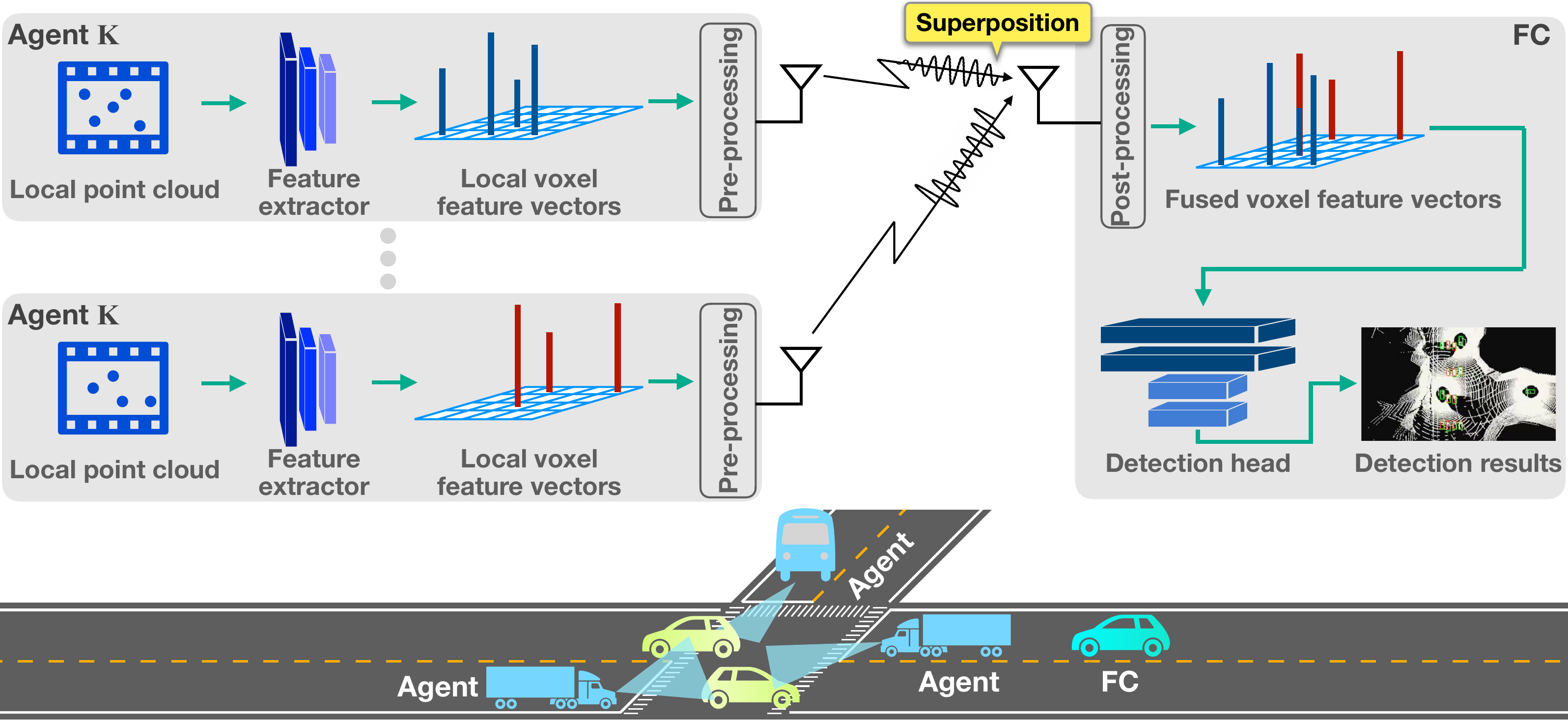}\label{fig_system}}
\hspace{0.5cm}
\hspace{-4mm}
\subfigure[]{\includegraphics[width=0.32\textwidth]{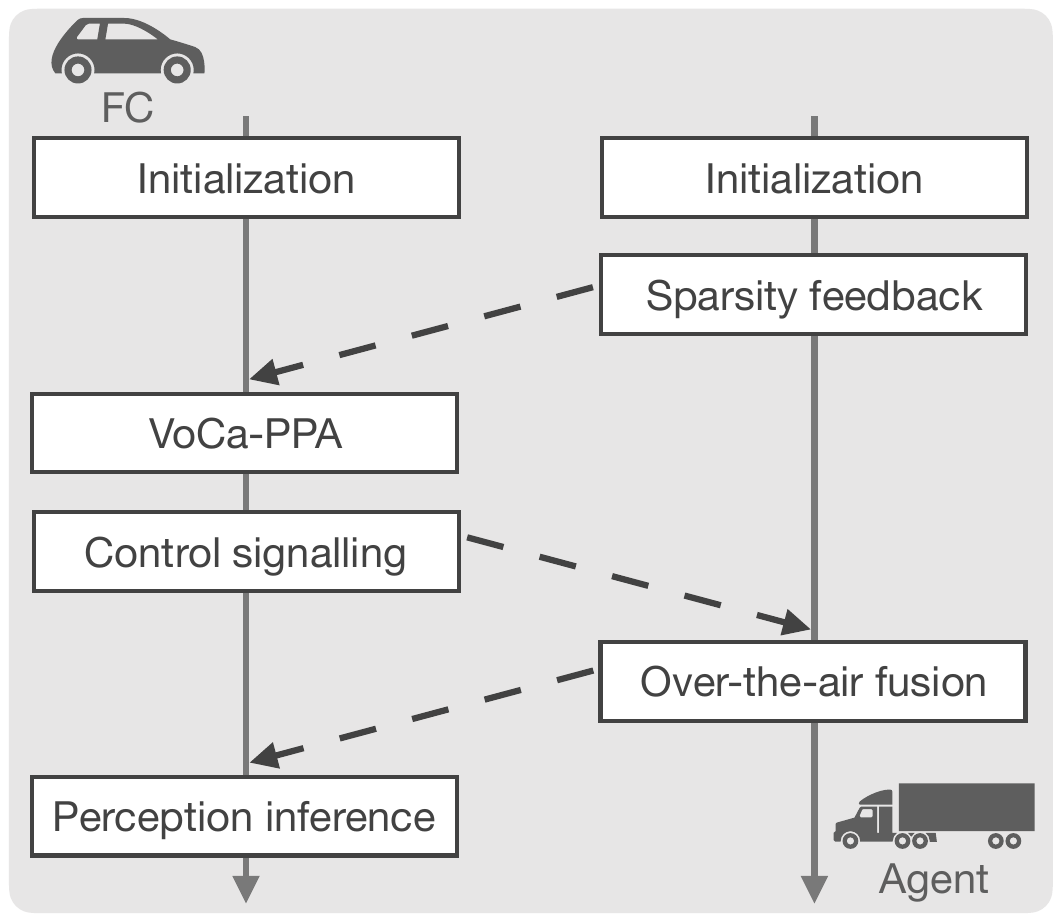}\label{fig_protocol}}
\vspace{-2mm}
\caption{(a) An ISEA system for environment perception in the context of autonomous driving. (b) Spatial AirFusion protocol.}
\label{fig_sysNprotocol}
\end{figure*}

\subsection{Agent Perception Model}
\label{subsec: agent_perception_model}

{ Each agent is equipped with a LiDAR or camera sensor that has its own perception range and perspective. Prior to fusion, each agent calibrates timestamp differences and performs local perspective transformation to project its view onto the FC's coordinates based on the relative position and speed using existing techniques such as coordinate offsets \cite{coordinate_offset} and AVR\cite{avr}. We thus assume a shared three-dimensional coordinate for all sensors, which is by convention partitioned into a regular grid with each cell referred to as a \emph{voxel}.} The numbers of partitions  along the depth, height, and width directions are denoted as $V_{\sf d}$, $V_{\sf h}$, and $V_{\sf w}$, respectively. Then the total number of voxels is given as $V = V_{\sf d} V_{\sf h} V_{\sf w} $, and the voxels are indexed by  $v=1, 2, ..., V$. As illustrated in Fig.~\ref{fig_system}, each agent utilizes its voxel-perception model to generate an $L$-dimensional feature vector for every voxel to capture the spatial object information contained within the voxel, termed  \emph{voxel feature vector}~\cite{Rukhovich2022WACV, Zhou2018CVPR}. For voxel $v$, its feature vector on agent $k$ is denoted as $\mathbf{f}_{k,v}\in \mathbb{R}^L$. {It can be a zero vector (i.e., $\mathbf{f}_{k,v}=\mathbf{0}$) if voxel $v$ is outside the perception range of agent or voxel $v$ is in the perception range but no objects are detected in voxel $v$ by agent $k$. Even in the latter case, the detection result may be false negative due to occlusion, noises or hardware imperfections of the agent's sensor. }

\subsection{Cooperative Sensing Model}
\label{subsec: wireless-cooperative-sensing-model}
The agents upload their voxel feature vectors, $\{\mathbf{f}_{k,v}\}_{1\leq k\leq K, 1 \leq v \leq V}$, to the fusion center over wireless links. Considering an arbitrary voxel, say voxel $v$, the result from fusing the associated  vectors is denoted as $\mathbf{g}_{v}$. {For two representative fusion functions, namely average-pooling and max-pooling, the $\ell$-th element of $\mathbf{g}_{v}$ is given as}
\begin{equation}
    \label{eqn: fusion-ground-truth}
    {g}_{v}[\ell] =\begin{cases}
        \frac{1}{K}\sum_{k=1}^{K} {f}_{k,v}[\ell], & \text{average pooling},\\
        \max_{1\leq k \leq K} {f}_{k,v}[\ell], &\text{max-pooling}.
    \end{cases} 
\end{equation}
Finally, the fusion center feeds the fused feature vectors, $\{\mathbf{g}_v\}_{v=1}^{V}$, into its perception model to obtain the perception results (e.g., object label).
{
\begin{remark}
    (Supported Aggregation Functions) \emph{In this paper, we have considered feature averaging or maximizing as the aggregation function, which can cover a majority of multi-sensor fusion schemes in cooperative perception for autonomous driving by up to a linear scaling at each sensor. For example, elementwise averaging/maximum is considered in F-Cooper and V2VNet, while weighted-sum fusion based on attention scores is adopted in Where2comm, BEVFormer and ActFormer (see the survey\cite{fusion_survey}). }
    \emph{Our framework is extensible to many other types of fusion functions. If the function belongs to the family of \emph{nomographic functions}, which includes, for example, square-root pooling and geometric mean, the extension is achieved by applying corresponding data pre- and post-processing at sensors and servers, respectively. For non-nomographic functions, approximations with nomographic functions can be designed by existing methods, e.g., \cite{nomographic_approx}. 
    }
\end{remark}
}

\subsection{Communication Model}
Spatial AirFusion wirelessly implements the above feature-fusion process over a broadband channel, which is modeled as follows.  The total bandwidth $B$ is partitioned into $M$ subcarriers using \emph{orthogonal frequency division multiplexing} (OFDM).  Without loss of generality, it is assumed that $M\geq V$, as otherwise transmission for all $V$ voxels can be carried out over multiple channel coherence blocks. The channel follows block fading where a subcarrier remains constant within a channel coherence block.  It is assumed that the channel coherence time is not shorter than the duration of $L$ symbol slots and that \emph{channel state information} (CSI) is available at the receiver and transmitters\footnote{As a common assumption in existing broadband AirComp literature (see, e.g., \cite{Wen2023TWC,Zhang2023TWC}), we assume reliable acquisition of CSI through downlink pilots and channel feedback via existing schemes such as frequency-domain interpolation\cite{channel_est_ofdm} and limited feedback\cite{limited_feedback}. While dedicated pilot design and feedback schemes (see, e.g., \cite{gu2023twc}) can further mitigate the overhead, relevant discussions are beyond the scope of this paper. }. This allows a voxel feature vector to be uploaded within a single channel coherence block.   
Assuming symbol-level synchronization (see \cite{GX2020TWC} for synchronization techniques), all agents simultaneously transmit their feature vectors on assigned subcarriers.  The $\ell$-th symbol received by the fusion center on the $m$-th subcarrier, $y_{m}[\ell]$, is  given by
\begin{equation}
    \label{eqn: channel_model}
    y_{m}[\ell]=\sum_{k=1}^K h_{k,m}p_{k,m}[\ell]x_{k,m}[\ell] + z_{m}[\ell], 
\end{equation}
where $x_{k,m}[\ell]$ denotes the $\ell$-th symbol transmitted by the $k$-th agent on the $m$-th subcarrier, $h_{k,m}$ the complex channel coefficient of subcarrier $m$ from  agent $k$ to the fusion center, $p_{k,m}[\ell]$ the precoding coefficient, and $z_{m}[\ell] \sim \mathcal{CN}(0,N_{0})$ the i.i.d. Gaussian noise with power $N_{0}$. Using training data, the symbols $\{x_{k,m}[\ell]\}$ can be normalized to be zero-mean and unit-variance on a long-term basis~\cite{Huang2023TWC}. 
Channel inversion precoding is adopted for magnitude alignment between received signals~\cite{Deniz2020TWC,Shi2020TWC}. The transmit power of agent $k$ on subcarrier $m$ is then given by $|p_{k,m}[\ell]|^2=\frac{P_{\mathsf{rx},m}[\ell]}{|h_{k,m}|^2},$ $\forall l$, where $P_{\mathsf{rx},m}[\ell]\geq 0$ denotes the receive SNR coordinated by the fusion center for the $\ell$-th symbol transmitted on subcarrier $m$. 
As the channel remains constant for all $\ell=1,2,\ldots,L$, we set $P_{\mathsf{rx},m}[\ell]\triangleq P_{\mathsf{rx},m}, \forall \ell$, and consequently $p_{k,m}[\ell]\triangleq p_{k,m}, \forall \ell$. 
Each agent limits the total transmission power per OFDM symbol to $P_{\sf max}$, which is given as
\begin{equation}
    \label{eqn: power-constraints}
    \sum_{m=1}^{M} \vert p_{k,m} \vert^2 \leq P_{\sf max}, \ \forall k.
\end{equation}

\subsection{Performance Metric}
The presence of channel distortion in Spatial AirFusion results in AirComp error, defined as the mean square error between the over-the-air aggregated data and the ground-truth fusion result\cite{GX2021WCM}. Under per-agent power constraints, AirComp error, known to be inversely proportional to the receive SNR, is dominated by the worst channel due to the required magnitude alignment of received signals via channel inversion\cite{Xiaowen2020TWC}. In the sensing context, the end-to-end sensing accuracy, prone to distortion in the aggregated intermediate features, has been shown to improve with the receive SNR in \cite{Huang2023TWC}. In AirFusion, we denote the receive SNR for the sub-task of aggregating voxel $v$'s features as $\gamma_v${, which controls the aggregation quality of $\bg_v$.} It is determined by the coordinated SNR level for its assigned subcarrier, i.e., $\gamma_v=P_{\mathsf{rx},m(v)}$ if subcarrier $m(v)$ is assigned for voxel $v$. { The performance metric shall thus be a function of the received SNR levels $\{\gamma_v\}_{v=1}^V$. To determine its form for sensing performance maximization requires a closer look into the downstream \emph{region proposal network} (RPN) \cite{rpn} for object detection tasks. An important property in RPN inference is its \emph{locality}. Specifically, RPN slides a small neural network over the aggregated feature map, which takes in a small spatial window of voxel features and outputs the detection results for the associated voxel. Mathematically, the detection result for voxel $v$, $\br_v$, is given by $\br_v=\text{RPN}\left(\{\bg_v\}_{v\in\mathcal{N}(v)}\right)$, where $\mathcal{N}(v)$ denotes the spatially neighboring voxels of voxel $v$. Then the object detection results for the entire space is obtained by sliding over all $v\in V$. We can see that the detection results for a given voxel $v$ only rely on a small subset of voxel features with spatial locality. In mission-critical tasks where misdetection in a single voxel can be catastrophic (e.g., missing a pedestrian in autonomous driving), reliable detection is demanded for all voxels. Hence, it important to rein in the feature distortion by improving $\gamma_v$ for \emph{every} voxel instead of simply controlling average distortion. We therefore propose to define the performance metric for Spatial AirFusion, denoted by $U$, as the minimum receive SNR across all voxels: $U=\min\limits_{v\in\{1,\ldots,V\}} \gamma_v$.}

\section{Spatial AirFusion Protocol and Operations}

The proposed Spatial AirFusion framework aims at efficiently aggregating multi-agent voxel features over a broadband channel, where the feature vectors on different agents but attributed to the same voxel are aggregated over a particular subcarrier. Targeting environment perception, Spatial AirFusion is differentiated from generic AirComp in that features exhibit heterogeneous sparsity across voxels due to diversified occlusion and finite detection ranges of agents, which is exploited for optimized resource allocation by VoCa Pairing and power control. The steps of the Spatial AirFusion protocol are illustrated in Fig.~\ref{fig_protocol} and detailed below.

\subsection{Sparsity Feedback} 
Assume that agents are synchronized in indexing voxels of the sensing region due to coordination by the fusion center (see Section~\ref{subsec: agent_perception_model}).
Voxel $v$ is called \emph{sparse} on agent $k$ if and only if the corresponding feature vector $\bff_{k,v}$ is a zero vector. Each agent calculates a binary sparsity vector $\mathbf{s}_{k}\in\{0,1\}^{V}$, $k=1,\ldots, K$, indicating the observed sparsity pattern of its voxels. Specifically, $\mathbf{s}_{k}[v]=0$ if voxel $v$ on agent $k$ is sparse and $\mathbf{s}_{k}[v]=1$ otherwise, i.e., 
\begin{equation}
    \bs_k[v] = 
    \begin{cases}
    1, & \Vert \mathbf{f}_{k,v} \Vert_0\geq 1, \\
    0, & \text{otherwise},
    \end{cases}
\end{equation}
where $\Vert \mathbf{f} \Vert_0$ is the vector zero-norm defined as the number of non-zero elements in $\mathbf{f}$. All agents report their sparsity vectors to the fusion center via a reliable control channel. The server assembles them into a sparsity pattern: $\mathbf{S}=\left[ \mathbf{s}_{1}, \mathbf{s}_{2}, \cdots, \mathbf{s}_{K} \right]^{T}$.
The entry on the $k$-th row and $v$-th column of matrix $\mathbf{S}$ is denoted as $S_{k,v}=\bs_k[v]$.

\subsection{Radio Resource Allocation} 

Given the sparsity pattern, $\mathbf{S}$, and transmit CSI, $\{h_{k,m}\}$, the server allocates subcarriers and transmit power for each agent. We denote $\mathbf{A}\in \{0,1\}^{V\times M}$ as the VoCa pairing matrix, where the $(v,m)$-th entry is given as
\begin{equation}
    A_{v,m} = 
    \begin{cases}
    1, & \text{subcarrier $m$ paired with voxel $v$,}\\
    0, & \text{otherwise}.
    \end{cases}
\end{equation}
To assure orthogonality between aggregations of all voxels, the following constraints are applied on assigning subcarriers:
\begin{equation}
    \label{eqn: pairting-orthgonality}
    \sum_{v=1}^{V} A_{v,m} \leq 1, \ \forall m=1, 2, \cdots, M.
\end{equation}
On the other hand, each voxel occupies exactly one subcarrier:
\begin{equation}
    \label{eqn: pairting-single}
    \sum_{m=1}^{M} A_{v,m} = 1, \ \forall v=1, 2, \cdots, V.
\end{equation}

Let $|p_{k,m}|^2$ denote the transmit power invested to subcarrier $m$ by agent $k$. Then, $\{\vert p_{k,m} \vert^2\}$ depend on the sparsity of paired voxels channel gains, and receive SNRs (after aggregation). { To be specific, agent $k$ does not participate in the aggregation of voxel $v$ if $S_{k,v}=0$, thereby setting its transmitting power to zero on the subcarrier designated for voxel $v$. This is mathematically given by: $p_{k,m}=0$ if $\sum_{v=1}^{V}S_{k,v}A_{v,m}=0$.}
All the agents participating in the transmission on the designated subcarrier shall set transmit power to align their signal magnitude as required for AirComp\cite{GX2020TWC}. It follows that the receive SNR, denoted as $\gamma_v$ for voxel $v$, is given as
\begin{equation}
    \label{eqn: receive-snr}
    \gamma_v = \sum_{m=1}^{M} A_{v,m} \frac{\vert p_{k,m}  h_{k,m} \vert^2 }{N_0}, \ \forall k\in\{k^{\prime} \vert S_{k^{\prime},v}=1\}.
\end{equation}
The above resource allocation decisions, $\mathbf{A}$ and $\{\gamma_v\}_{v=1}^V$, are broadcast to all agents such that each onboard agent sets its precoding coefficients accordingly as follows:
\begin{eqnarray}
    \label{eqn: pre-coding}
    p_{k,m} &=& \frac{\sqrt{N_0}}{h_{k,m}}\sum_{v=1}^{V} \sqrt{\gamma_{v}}S_{k,v} A_{v,m}.
\end{eqnarray}
The control of resource allocation, i.e., VoCa-PPA, is optimized in the subsequent sections.

\subsection{Over-the-Air Fusion}
All agents simultaneously transmit their voxel features using assigned subcarriers and power levels as specified in $\mathbf{A}$ and $\{p_{k,m}\}$. Consider an arbitrary agent $k$ and an arbitrary symbol $\ell$. { Assume that average pooling is the desired fusion function. Then, the feature pre-processing is implemented by normalizing the $\ell$-th feature coefficient of voxel $v$ on agent $k$, $f_{k,v}[\ell]$, yielding the pre-processed feature $\tilde{x}_{k,v}[\ell]$ as given by}
\begin{equation}
    \label{eqn: pre-processing}
    \tilde{x}_{k,v}[\ell] = \frac{1}{\sigma}\left(f_{k,v}[\ell]-\mu\right),
\end{equation}
{where the normalization parameters $\sigma$ and $\mu$ in~\eqref{eqn: pre-processing} are shared by all agents and set such that the distribution of pre-processed features is zero-mean and unit-variance.
The extension to other fusions functions (e.g., max-pooling~\cite{Huang2023TWC}) is straightforward by applying additional post- and/or pre-processing functions.}
The pairing matrix $\mathbf{A}$ maps the pre-processed features, $\{\tilde{x}_{k,v}[\ell]\}_{v=1}^V$ to the $\ell$-th symbol of each subcarrier. Then the symbol transmitted by agent $k$ over subcarrier $m$ can be written as\footnote{For notational simplicity, we have assumed that the features are modulated onto the real, or in-phase, component of the transmitted symbols, which is also common in many AirComp literature, e.g., \cite{Eldar2021TSP,Zhang2023TWC} In practice, it is possible to modulate features on both the in-phase and quadrature components, which reduces the AirComp latency by half (see, e.g., \cite{xu_dpca,wen2024twc}). The extension of our work to this case is straightforward without changing the design of our framework and algorithms.}
\begin{equation}
    \label{eqn: map-to-symbol}
    x_{k,m}[\ell] = \sum_{v=1}^{V} A_{v,m} \tilde{x}_{k,v}[\ell].
\end{equation}

Combining \eqref{eqn: pre-coding} and \eqref{eqn: map-to-symbol} with the AirComp operation in \eqref{eqn: channel_model} yields the $\ell$-th symbol received at the fusion center:
\begin{align}
    y_{m}[\ell] &= \sum_{k=1}^{K} \left[ \left( \sum_{v=1}^{V}\sqrt{N_0\gamma_v} S_{k,v} A_{v,m} \right) x_{k,m}[\ell]\right] + z_{m}[\ell], \nonumber \\
    &= \sum_{v=1}^{V}\left[ A_{v, m} \sqrt{N_0\gamma_{v}} \left( \sum_{k=1}^K S_{k,v}\tilde{x}_{k,v}[\ell]\right) \right] + z_{m}[\ell].
    \label{eqn: received-symbol}
\end{align}
The post-processing operation (i.e., denormalization and real-part extraction) is then designed to compute the estimated fused feature vector for voxel $v$, $\tilde{\bg}_v\in\mathbb{R}^L$, such that its $\ell$-th element is given as
\begin{eqnarray}
    \tilde{g}_{v}[\ell] 
    &=& 
    \Re\left[  \sum_{m=1}^{M}\left[\frac{A_{v,m}}{K\sqrt{N_0 \gamma_{v}}}\left(\sigma y_{m}[\ell] +\sum_{k=1}^{K}\mu S_{k,v}  \right)\right]\right], \nonumber\\
    &=& g_{v}[\ell] + \Re\left[ \frac{\sigma\sum_{m=1}^M A_{v,m} z_{m}[\ell]}{K\sqrt{N_0 \gamma_{v}}}\right].
    \label{eqn: fused-feature-coefficient}
\end{eqnarray}
It follows that the vector can be expressed in terms of its ground-truth given in~\eqref{eqn: fusion-ground-truth} as
\begin{equation}
    \label{eqn: receive-truth-noise}
    \tilde{\mathbf{g}}_{v} = {\mathbf{g}}_{v} + \mathbf{w}_{v},
\end{equation}
where $\mathbf{w}_{v}$ is a vector of i.i.d. Gaussian noise variables following $\mathcal{N}\left(0, \frac{1}{2} K^{-2}\gamma_v^{-1}N_0^{-1}\sigma^2\right)$. Last, the fusion center assembles all the fused voxel feature vectors, $\{\tilde{\bg}_{v}\}$, and feeds them into the downstream perception head to obtain the final inference results.
{
\begin{remark}
    (System Scalability) \emph{One key advantage of Spatial AirFusion against digital orthogonal access is its high scalability w.r.t. the number of agents and data volume. Aligned with real-world challenges, increasing the number of agents in cooperative perception is a trend in relevant literature, e.g., from $2$ agents in F-Cooper\cite{fcooper} to $6$ agents in OPV2V\cite{xu2022opencood} to $12$ agents in V2X-Sim\cite{v2xsim}, which results in growing communication overhead for digital orthogonal access. In contrast, the increase in number of agents does not add to latency or bandwidth consumption in AirFusion thanks to simultaneous access, but also mitigates both channel and data noise as found in \cite{Chen2023arxiv}. The increase in data volume is due to sensor advancements, e.g., LiDAR sensing resolution and range. Spatial AirFusion copes with this challenge by 1) fusion of resolution-invariant voxel features instead of raw data; 2) fusion on sparse non-empty voxels instead of all voxels in the sensing region. }
\end{remark}
}

\section{VoCa-PPA: Problem Formulation}
\label{sec: optimization}
Recall that the VoCa-PPA problem of Spatial AirFusion aims at allocating subcarriers and transmit power to agents/voxels so as to maximize the minimum receive SNR among voxels, which is formulated as follows. Given the pairing constraints~\eqref{eqn: pairting-orthgonality} and \eqref{eqn: pairting-single} and by substituting the channel inversion~\eqref{eqn: pre-coding} into the instantaneous power constraints in \eqref{eqn: power-constraints}, the optimization problem can be formulated as
\begin{equation*}\text{(P1)}\quad 
        \begin{aligned}
        \max\limits_{\mathbf{A},\{\gamma_v\}_{v=1}^{V}}&\min\limits_{v}\quad \, \gamma_v \\
        \mathrm{s. t.}\quad & 
        A_{v,m} \in \{0, 1\}, \ \forall v, m, \\
        &\sum_{v=1}^{V}A_{v,m} \leq 1, \ \forall m, \quad\sum_{m=1}^{M}A_{v,m} = 1, \ \forall v,\\
        &\sum_{m=1}^{M}\frac{N_0}{{|h_{k,m}|^2}}\sum_{v=1}^{V}S_{k,v}A_{v,m}\gamma_{v}\leq P_{\sf max}, \ \forall k.
        \end{aligned}
\end{equation*}
Problem P1 is a mixed-integer programming problem. To simplify it, we derive the optimal receive SNRs as functions of the pairing matrix $\mathbf{A}$, shown in the following lemma. Its proof is by a standard transformation of the power allocation problem given $\bA$ into a linear program and solving it via Lagrange duality and thus omitted for brevity.
\begin{Lemma}[\emph{Optimal Power Allocation}]
    \label{theorem-snr-conditioned-on-A}
    \emph{Given the VoCa pairing matrix $\mathbf{A}$, setting an equal SNR level across all voxels, i.e., $\gamma_v = \gamma^{*}(\mathbf{A})$ for all $v$, is optimal for Problem P1, where $\gamma^{*}(\mathbf{A})$ is given as
    \begin{equation}
        \label{eqn: optimized_snr}
        \gamma^{*}(\mathbf{A})=P_{\sf max} \left( \max\limits_{k} N_0\sum_{v=1}^{V} \sum_{m=1}^{M}\frac{S_{k,v}A_{v,m}}{|h_{k,m}|^{2}} \right)^{-1}.
    \end{equation}
    Substituting $\gamma^{*}(\mathbf{A})$ into~\eqref{eqn: pre-coding} yields the optimal transmit power of each agent over a subcarrier,}
    \begin{equation}
        \label{eqn: snr-to-power}
        p^{*}_{k,m}(\mathbf{A}) = \frac{\sqrt{N_{0}\gamma^{*}(\mathbf{A})}}{h_{k,m}}\sum_{v=1}^{V}S_{k,v}A_{v,m}.
    \end{equation}
\end{Lemma}
It can be observed from \eqref{eqn: optimized_snr} that the achievable SNR levels depend on a bottleneck agent characterized by weakest overall channel conditions by considering all voxels and subcarriers.
Without compromising its optimality, Problem P1 can be simplified by substituting~\eqref{eqn: optimized_snr} into the objective. This leads to the following equivalent VoCa Pairing problem:
\begin{equation*}\text{(P2)}\quad 
        \begin{aligned}
        \min\limits_{\mathbf{A}}\quad &\max\limits_{k}\quad \, \sum_{v=1}^{V} \sum_{m=1}^{M}c_{k,m}S_{k,v}A_{v,m}\triangleq F(\bA) \\
        \mathrm{s. t.}\quad & 
        A_{v,m} \in \{0, 1\}, \ \forall v, m, \\
        & \sum_{v=1}^{V}A_{v,m} \leq 1, \ \forall m, \quad\sum_{m=1}^{M}A_{v,m} = 1, \ \forall v,
        \end{aligned}
\end{equation*}
where the constant $c_{k,m}\triangleq \frac{N_0}{|h_{k,m}|^{-2}}$. 
This is a combinatorial optimization problem with a max-linear objective, which is known to be NP-hard in general\cite{ehrgott2006discussion}. A set of algorithms are designed in the following sections to overcome this challenge. 

\section{Greedy VoCa-PPA Algorithm}
\label{subsec: greedy-pairing}
In this section, we first develop a low-complexity solution to Problem P2 for VoCa Pairing based on a greedy heuristic. Then, combining the greedy algorithm and the optimal power allocation scheme yields the greedy VoCa-PPA algorithm for Spatial AirFusion control. 
\subsection{Greedy VoCa Pairing}

The proposed greedy pairing algorithm in principle sequentially pairs a single voxel with the locally optimal subcarrier. The specific algorithm is designed as follows.
\begin{itemize}
    \item \textbf{Initialization.} The pairing matrix $\mathbf{A}$ is initialized as $\mathbf{A} \leftarrow \mathbf{0}^{V\times M}.$
    \item \textbf{Iteration.} In each iteration, say the $v$-th one, $\mathbf{A}$ is updated in a greedy manner, i.e., upon solving an optimization problem that seeks the best subcarrier for the $v$-th voxel. Specifically,  only the $v$-th voxel is addressed in this iteration. Dropping other voxels in Problem P2 yields the greedy optimization problem for voxel $v$ as
\begin{equation*}\text{(P3)} 
        \begin{aligned}
        \min\limits_{\{A_{v,m}\}_{m=1}^{M}}\quad &\max\limits_{k}\quad \, \sum_{m=1}^{M}{c_{k,m}}S_{k,v}A_{v,m} \\
        \mathrm{s. t.}\quad & 
        A_{v,m} \in \{0, 1\}, \ \forall m, \\ 
        &\sum_{n=1}^{v}A_{n,m} \leq 1, \ \forall m,\quad\sum_{m=1}^{M}A_{v,m} = 1.
        \end{aligned}
\end{equation*}
In Problem P3, only the pairing parameters for voxel $v$ are optimized while the others are fixed. The optimal solution to Problem P3, $\big\{A^{*}_{v,m}\big\}_{m=1}^{M}$, can be easily obtained for any given $v$ as,
\begin{equation}
    \label{eqn: greedy_pairing}
    A^{*}_{v,m} = 
    \begin{cases}
    1,&m = \mathop{\arg\min}\limits_{m\in\{m^{\prime}\vert \sum_{n=1}^{v-1}A_{n,m^{\prime}}=0\}} \max\limits_{k} \ {c_{k,m}}{S_{k,v}} ,\\
    0,&\text{otherwise}.
    \end{cases}
\end{equation}
To end the $v$-th iteration, the entries specifying pairing of voxel $v$ in $\bA$ are updated as ${A}_{v,m}\leftarrow{A}^{*}_{v,m}$ for all $m$.
\end{itemize}

\textbf{Sequence optimization.} An optimized sequence of voxels in greedy pairing can boost the performance, i.e., improve the achieved voxel-level receive SNRs. To this end, we first propose a metric for sorting the voxels. One voxel can differ from another in the level of sparsity, i.e., the number of agents participating in aggregation. We refer to the voxels involving a small number of agents as \emph{high-sparsity} voxels, in comparison against the \emph{low-sparsity} voxels involving a large number of agents. Intuitively, the latter should be assigned subcarriers with favorable channel conditions as it is well-known in the AirComp literature that the receive SNR decreases as more agents participate\cite{Vucetic2020TWC}. Based on this principle, we propose a sparsity-aware permutation strategy that prioritizes low-sparsity voxels in greedy pairing. The permutation function $\pi(\cdot)$ maps an arbitrary entry $v$ in the set $\{1, 2, ..., V\}$ to its image $\pi(v)$, which determines the index of iteration in the greedy pairing algorithm. Specifically, $\pi(\cdot)$ is constructed via sorting the sequence $1, 2, ..., V$ in descending order of their sparsity levels $\sum_{k=1}^{K}S_{k,v}$. This yields a sorted sequence $\pi(1), \pi(2), ..., \pi(V)$, where we place $v_1$ before $v_2$ in the case of $\sum_{k=1}^{K}S_{k,v_1}=\sum_{k=1}^{K}S_{k,v_2}$ if $v_1<v_2$. It can be easily verified that the constructed $\pi(\cdot)$ is an \emph{bijective} function and prioritizes low-sparsity voxels. 

\subsection{Greedy VoCa-PPA}
The control algorithm, named greedy VoCa-PPA, combines the above greedy pairing with an optimized sequence and the power allocation scheme in Lemma~\ref{theorem-snr-conditioned-on-A}, which is summarized in Algorithm~\ref{algo: greedy_control}. Its input $\mathbf{H}$ is a $K$-by-$M$ matrix of channel gain with $h_{k,m}$ being its entry in the $k$-th row and $m$-th column. As a remark, the control signalling in Algorithm~\ref{algo: greedy_control} involves broadcasting a \emph{sparse} and \emph{binary} matrix $\mathbf{A}^{\dagger}$ and a scalar $\gamma^{*}\left(\mathbf{A}^{\dagger} \right)$. The former of the two control parameters can be easily encoded into $\log_2\left(\frac{M!}{(M-V)!}\right)\leq V\log_2(M)$ bits while the latter can be quantized into $32$ bits following the floating-point precision convention. The signalling thus can be implemented over a downlink feedback channel with its overhead neglected.
{
\subsection{Complexity Analysis}
The time complexity of Algorithm~1 is presented as follows. Before starting the iteration, we sort $\{c_{1,m},\ldots,c_{K,m}\}$ in descending order for each $m=1,\ldots,M$ and store the results. The complexity of this step is $\mathcal{O}(MK\log K)$. Then, in the $v$-th iteration, \eqref{eqn: greedy_pairing} shall be evaluated. The inner $\max\limits_{k} \ {c_{k,m}}{S_{k,v}}$ is obtained with $\mathcal{O}(1)$ given $\{c_{1,m},\ldots,c_{K,m}\}$ sorted and $S_{k,v}$ binary, and the outer operation costs $\mathcal{O}(M)$. Hence, the complexity of the iterating process is $\mathcal{O}(MV)$. The total complexity of Algorithm~1 is thus $\mathcal{O}(M\max\{K\log K,V\})$.

This complexity is comparable to basic algorithms in OFDM such as \emph{Fast Fourier Transform} (FFT). Moreover, our algorithm does not involve floating-point multiplications, making it highly efficient for implementations in standard hardware\cite{hardware_testbed}.
}
\begin{algorithm}[t]
\caption{Greedy VoCa-PPA}
\label{algo: greedy_control}
\textbf{Input:} Sparsity matrix $\mathbf{S}$ and channel matrix $\mathbf{H}$\;
\textbf{Prioritization:} Determine $\pi(\cdot)$ as given in Section~\ref{subsec: greedy-pairing}\;
\textbf{Initialization:} $\mathbf{A}^{\dagger} =\mathbf{0}$\;
\textbf{for} $v=1,2,\cdots,V$ \textbf{do (greedy pairing)}\\
\ \ \ \  Evaluate $A^{\dagger}_{\pi(v),m}$ for $m=1, 2, ..., M$ by~\eqref{eqn: greedy_pairing} \;
\textbf{Setting SNR:} Substitute $\mathbf{A}^{\dagger}$ into~\eqref{eqn: optimized_snr} for $\gamma^{*}\left(\mathbf{A}^{\dagger} \right)$\;
\textbf{Signalling:} Broadcast the control parameters $\mathbf{A}^{\dagger}$, $\gamma^{*}\left(\mathbf{A}^{\dagger}\right)$ to agents, which then set their transmit power by~\eqref{eqn: snr-to-power}\;
\end{algorithm}

\section{Optimal VoCa-PPA: Compact Tree Design}
\label{sec: optimal_paining}
The greedy VoCa-PPA algorithm in the preceding section is computation-efficient but sub-optimal. In this and the next sections, we present an optimal and efficient approach for solving the VoCa-PPA Problem in P1 or equivalently Problem P2. The tree-search based approach consists of two components -- compact tree design in this section and fast tree search in the next section. In general, Problem P2 can be viewed as a special case of the max-linear assignment problem, and its optimal solution can be searched for using the well-known ranking method (see, e.g., \cite{BELHOUL201497}). The novelty of our design, which yields a higher efficiency than the existing method, lies in exploiting the special structure of Problem P2. In particular, a derived useful property of its objective leads to a dramatic reduction of the dimensionality of the search space. {The motivation of organizing the search into a search tree is to reduce the search complexity by node pruning with a branch-and-bound-inspired method. The method hinges on proper selection of the branching variable and bounding the global objective with local ones, as will be introduced shortly.}  As a common practice in solving bipartite matching problems, assume equal numbers of voxels and subcarriers $M=V$ in the sequel without loss of generality as the case of $M>V$ can be augmented with $(M-V)$ dummy voxels with all-zero sparsity indicators for all agents.

\subsection{A Useful Property of Objective Function}

Consider the objective function of Problem P2, $F(\bA)=\max_k f_k(\bA)$, where $f_{k}(\bA)\triangleq \sum_{v=1}^V\sum_{m=1}^{M}c_{k,m}S_{k,v}A_{v,m}$.
To facilitate exposition, let $m(v)$ denote the index of the unique non-zero entry in the $v$-th row of $\mathbf{A}$, which indicates that the $v$-th voxel is mapped to the $m(v)$-th subcarrier. We thus have $m(v_i)\neq m(v_j)$ when $v_i\neq v_j$. Then, $f_{k}(\mathbf{A})$ can be rewritten as
\begin{equation}\label{eqn: single-sensor-obj}
\begin{aligned}
    f_{k}(\mathbf{A}) &= \sum_{v=1}^V S_{k,v} c_{k,m(v)}= \sum_{v\in \mathcal{V}_k} c_{k,m(v)} = \sum_{m\in \mathcal{M}_{k}} c_{k,m},
\end{aligned}
\end{equation}
where $\mathcal{V}_{k} = \{v|S_{{k},v}=1\}$, the index set of all non-sparse voxels for agent ${k}$, and $\mathcal{M}_{k} = \{m(v)\}_{v\in \mathcal{V}_{k}}$, the set of subcarriers selected for non-sparse voxels, is the image of $\mathcal{V}_{k}$ under the mapping $m(\cdot)$ with $|\mathcal{M}_{k}| = |\mathcal{V}_{k}|$. An important observation is that $f_{k}$ depends only on the set $\mathcal{M}_{k}$ but not the specific one-to-one mappings. As a result, for two voxels $v_1$ and $v_2$ both in $\mathcal{V}_{k}$ or $\mathcal{V}\setminus\mathcal{V}_{k}$ (or equivalently having identical sparsity indicators $S_{k,v_1}=S_{k,v_2}$ on agent $k$) swapping their associated subcarriers does not alter the value of $f_k$ as $\mathcal{M}_k$ remains unchanged. This argument can be extended from $f_k$ to the objective $F(\bA)$ since it is a function of $\{f_k(\bA)\}$. Specifically, consider the case that two voxels $v_1$ and $v_2$ have identical sparsity indicators for all $K$ agents, or in other words, the two voxels have exactly the same sparsity vector, i.e., $\bt_{v_1}=\bt_{v_2}$, where $\mathbf{t}_{v}$ is the $v$-th column of the sparsity pattern matrix $\mathbf{S}$. Then exchanging their assigned subcarriers does not change the objective value. In such cases, we call the two voxels \emph{homogeneous} due to their equivalence in subcarrier assignment. Aggregating all voxels which are homogeneous to each other results in the concept of a \emph{homogeneous subset}, denoted by $\mathcal{H}(\br^q)\triangleq \mathcal{H}^q$, where $\br^q\in\{0,1\}^K$ indicates the sparsity vector shared by all voxels in $\mathcal{H}^q$. Mathematically, for all $v\in\mathcal{H}^q$, $\bt_{v}=\br^q$. As $\br^q$ is a binary vector with length $K$, it has at most $2^K$ possibilities, as indexed by $q=1,\ldots,2^K$. The above property is stated formally in the following lemma.

\begin{Lemma}\label{lemma: homogeneous_voxels}
    \emph{Consider a VoCa Pairing $m(\cdot):\mathcal{V}\rightarrow\mathcal{M}$ and two voxels in the same \emph{homogeneous subset} $v_1,v_2\in\mathcal{H}^q$. A new pairing $m'(\cdot)$ with $m'(v_1)=m(v_2)$, $m'(v_2)=m(v_1)$ while $m(v)=m'(v)$ for all $v\neq v_1,v_2$ yields the same objective value as $m(\cdot)$.}
\end{Lemma}
The above lemma suggests that once the mapping between a homogeneous subset of voxels to an equal-size subcarrier subset is determined, the element-wise mapping can be arbitrary without altering the objective value. The property is the fundamental reason for the efficiency of the proposed solution approach.
\subsection{Compact Solution Space}
The property in Lemma~\ref{lemma: homogeneous_voxels} is exploited in the sequel to define a compact solution space comprised of subset-to-subset mappings, which features much lower dimensionality as opposed to the original space of all possible one-to-one mappings $m(\cdot): \mathcal{M}\rightarrow\mathcal{V}$. 

To begin with, relevant terminologies are introduced as follows. Let $\{\mathcal{P}^j\}_{j=1}^{N(\phi)}$ be a non-overlapping partition of the voxel set $\mathcal{V}$ with $\bigcup_{j=1}^{N(\phi)}\mathcal{P}^j=\mathcal{V}$ and $\mathcal{P}^i\cap\mathcal{P}^j=\emptyset$ for any $i\neq j$, where $1\leq N(\phi) \leq V$ is the number of disjoint subsets.  A subset-to-subset mapping $\phi$ with $\mathrm{dom}(\phi)=\{\mathcal{P}^j\}_{j=1}^{N(\phi)}$ pairs $\mathcal{P}^j$ with $\phi(\mathcal{P}^j)$ for $j=1,\ldots,N(\phi)$
where $\{\phi(\mathcal{P}^j)\}_{j=1}^{N(\phi)}$ is required to be a non-overlapping partition of $\mathcal{M}$, i.e., $\bigcup_{j=1}^{N(\phi)}\phi(\mathcal{P}^j)=\mathcal{M}$ and $\phi(\mathcal{P}^i)\cap\phi(\mathcal{P}^j)=\emptyset$  for any $i\neq j$. In addition, equal sizes are set for a voxel subset and its paired subcarrier subset, as given by $|{\mathcal{P}^j}|=|{\phi(\mathcal{P}^j)}|$ for all $j$.
A bijective mapping, $m(v)$, \emph{satisfies} $\phi$ if and only if for any $v$, $v\in \mathcal{P}^j$ leads to $m(v)\in \phi(\mathcal{P}^j)$. In this sense, $\phi$ encompasses all bijective mappings between $\mathcal{V}$ and $\mathcal{M}$ that maps ${\mathcal{P}}^j$ exactly to $\phi(\mathcal{P}^j)$.

To completely determine the objective function of Problem P2 requires a subset-to-subset mapping $\phi_\mathrm{sol}$ with $\mathrm{dom}(\phi_\mathrm{sol})=\{\mathcal{H}^j\}_{j=1}^{2^K}$, which specifies the mapped subcarrier subset for any homogeneous voxel subsets, say, $\mathcal{H}^j$, as $\phi_\mathrm{sol}(\mathcal{H}^j)$.
    {Denote the set of all bijective mappings that satisfy $\phi_\mathrm{sol}$ as $\mathcal{C}(\phi_\mathrm{sol})$, which by Lemma~\ref{lemma: homogeneous_voxels} yield the same objective value.} Note that the union of $\mathcal{C}(\phi_\mathrm{sol})$ for all possible $\phi_\mathrm{sol}$ covers exactly the original solution space. It is therefore equivalent to consider the reduced-dimension space of $\phi_\mathrm{sol}$ as the solution space of Problem P2. The dimensions of the new solution space are determined by the number of possibilities of disjoint set partitions $\{\phi_\mathrm{sol}(\mathcal{H}^1),\phi_\mathrm{sol}(\mathcal{H}^2),\ldots,\phi_\mathrm{sol}(\mathcal{H}^{2^K})\}$ with $\bigcup_{j=1}^{2^K}\phi_\mathrm{sol}(\mathcal{H}^{j})=\mathcal{M}$ and the size of each subset fixed as $|\phi_\mathrm{sol}(\mathcal{H}^{j})|=|\mathcal{H}^j|$, which is calculated as $\frac{M!}{|\mathcal{H}_1|!|\mathcal{H}_2|!\cdots |\mathcal{H}_{2^K}|!}$. Thereby, we can achieve complexity reduction by orders of magnitude as compared with the original solution space, which encompasses all possible mappings between $\mathcal{M}$ and $\mathcal{V}$ and therefore has a size of $M!$.

\subsection{Tree Construction}
Finding the optimal solution to Problem P2 can be achieved by an enumeration of the compact solution space defined in the preceding subsection, which is still exponential in $M$ due to the suggested NP-completeness of Problem P2. We propose to organize the solution enumeration into a tree search. A naive approach to tree construction would be to sequentially branch on the selection of subsets $\phi_\mathrm{sol}(\mathcal{H}^{j})$, but this method is unlikely to benefit from complexity reduction by node pruning. Instead, our approach is to sequentially branch on the local objective $f_k(\bA)$ by assigning subcarriers to certain groups of homogeneous subsets identified by the sparsity indicators of the currently considered agent, which underpins the efficient tree-search algorithm with node pruning in Section~VI-A. In the sequel, we index the $K$ agents sequentially from $1$ to $K$. However, such an agent ordering can be arbitrary, which affects not the optimality but the empirical complexity. In this aspect, an agent-ordering algorithm is presented in the next section.  The search tree is illustrated  in Fig.~\ref{fig_tree_search}. For a general node $w$, let $d(w)$ denote its depth, i.e., the length of the (shortest) path connecting it to the root node. The maximum depth of the search tree equals $K$, and a node with depth $K$ is defined as a \emph{leaf node}. 
\subsubsection{Branching on Local Objectives}

 To begin with, we discuss the branches of the root node $w_0$ with depth $0$, i.e., its set of child nodes with depth $1$, by analyzing possible local objectives for agent $1$. Recall that for agent $1$, its associated local objective, $f_{1}(\bA)$ in \eqref{eqn: single-sensor-obj}, is fully determined by the subcarrier set assigned to agent $1$'s non-sparse voxels $\mathcal{V}_{1}$. Let $r^q_k$ be the $k$-th element of $\br^q$, indicating whether the homogeneous subset $\mathcal{H}^q$ is sparse on agent $k$. We can express $\mathcal{V}_{1}$ as $\mathcal{V}_{1}=\bigcup\{\mathcal{H}^q|r^q_1=1\}\triangleq \mathcal{U}_1$, i.e., the union set of homogeneous subsets which are non-sparse on agent ${1}$, and similarly $\mathcal{V}\setminus\mathcal{V}_{1}=\bigcup\{\mathcal{H}^q|r^q_1=0\}\triangleq \mathcal{U}_0$. This can be interpreted as dividing all homogeneous sets into two groups according to the sparsity on agent ${1}$.  The local objective values $f_{1}$ are then determined by $\phi_1$ with $\mathrm{dom}(\phi_1)=\{\mathcal{U}_1, \mathcal{U}_0\}$, which characterizes the assignment of subcarriers between non-sparse and sparse voxels of agent $1$, and mathematically given by 
 \begin{equation}\label{eqn: local_obj_rho_1}
     f_{1}(\phi_1)=\sum_{m\in \phi_1(\mathcal{U}_1)}c_{1,m}.
 \end{equation} 
 The number of all possible $\phi_1$, which generate (generally) distinct local objective values $f_{1}$, is equal to the number of size-$|\mathcal{U}_0|$ subsets of $\mathcal{M}$, i.e., $N(w_0)=\frac{M!}{|\mathcal{U}_0|!(M-|\mathcal{U}_0|)!}$. Each of the possible $\phi_1$ is represented by one child node of the root node, $w_j$, where $j=1,\ldots,N(w_0)$. The node set $\{w_1,\ldots,w_{N(w_0)}\}$ constitute all branches, or child nodes of the root node $w_0$. 
 
     \begin{figure}[t]
    \centering
\includegraphics[width=0.48\textwidth]{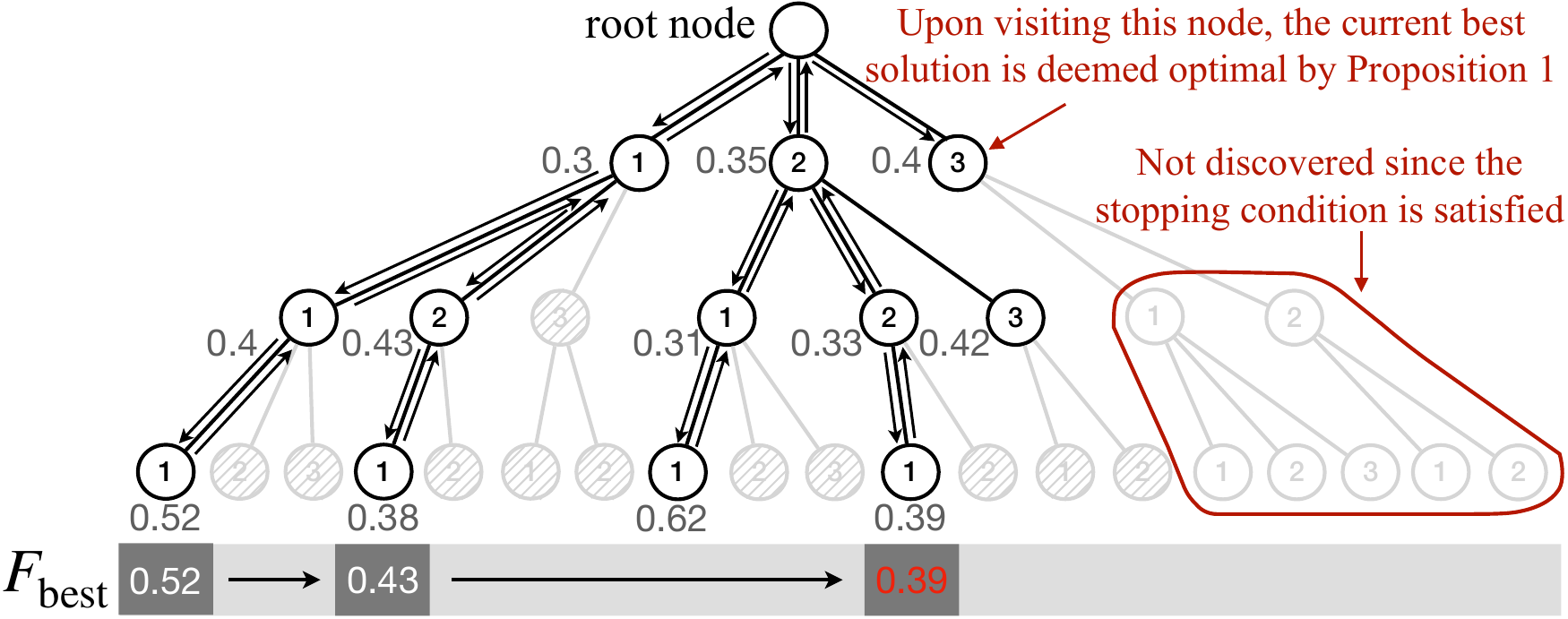}
\vspace{-2mm}
    \caption{An example of a search tree for the optimal solution of Problem P2, with maximum depth, i.e., the number of agents, $K=3$. Nodes pruned by Proposition~2 are marked with strides. }
    \label{fig_tree_search}
\end{figure}

Consider an arbitrary node, say $w_j$, and its associated subset-to-subset mapping is denoted as $\phi_1^{w_j}$ with $\mathrm{dom}(\phi_1^{w_j})=\{\mathcal{U}_1, \mathcal{U}_0\}$. While $\phi_1^{w_j}$ fixes agent $1$'s local objective value to $f_{1}(\phi_1^{w_j})$ by \eqref{eqn: local_obj_rho_1}, it only specifies the image of $\mathcal{U}_1, \mathcal{U}_0$, which are unions of homogeneous subsets, rather than each of $\{\mathcal{H}_j\}$, resulting in under-determined values for $f_{j}$, $j>1$. We thus aim to further subdivide the current mapping,  $\phi_1^{w_j}$, to a finer granularity by considering the sparsity pattern of the next agent $2$ such that $f_{2}$ is determined while $f_{1}$ fixed as $f_{1}(\phi_1^{w_j})$. Since $f_{2}$ depends on the image of $\bigcup\{\mathcal{H}^q|r^q_2=1\}$ and $\bigcup\{\mathcal{H}^q|r^q_2=0\}$, to determine both $f_{1}$ and $f_{2}$ requires mapping each of $\{\mathcal{U}_{11},\mathcal{U}_{10},\mathcal{U}_{01},\mathcal{U}_{00}\}$ to a subcarrier subset, where $\mathcal{U}_{b_1 b_2}\triangleq\bigcup\{\mathcal{H}^q|r^q_1=b_1,r^q_2=b_2\}$. Such a mapping is denoted as $\phi_2$ with $\mathrm{dom}(\phi_2)=\{\mathcal{U}_{11},\mathcal{U}_{10},\mathcal{U}_{01},\mathcal{U}_{00}\}$. On the other hand, conditioning on $f_{1}=f_{1}(\phi_1^{w_j})$ requires $\phi_2(\mathcal{U}_{11})\cup \phi_2(\mathcal{U}_{10})=\phi_1^{w_j}(\mathcal{U}_{1})$ and $\phi_2(\mathcal{U}_{01})\cup \phi_2(\mathcal{U}_{00})=\phi_1^{w_j}(\mathcal{U}_{0})$. Under the above condition, the possible outcomes of $f_{2}$ while fixing $f_1$ constitute all possible branches of $w_j$. This branching procedure can be recursively applied until reaching a leaf node, which determines all $\{f_j\}_{j=1}^K$. In the sequel, the branching procedure for a general node is presented. 
 
\subsubsection{General Nodes}
Consider a general node $w$. The steps to discover its child nodes are as follows. The node $w$ with depth $d(w)$ represents a partial solution to Problem P2 characterized by a subset-to-subset mapping $\phi_{d(w)}^w$. It domain is given as $\mathrm{dom}(\phi_{d(w)}^w)=\{\mathcal{U}_{b_1 b_2\cdots b_{d(w)}}\}_{b_i\in\{0,1\}}$, where $\mathcal{U}_{b_1 b_2\cdots b_{d(w)}}\triangleq\bigcup\{\mathcal{H}^q|r^q_{1}=b_1, \ldots,r^q_{d(w)}=b_{d(w)}\}$. The local objectives $f_1,\ldots,f_{d(w)}$ are determined by $\phi_{d(w)}^w$, as given by
 \begin{equation}\label{eqn: local_obj_rho_w}
     f_{j}(\phi_{d(w)}^w)=\sum_{m\in\bigcup_{b_{j}=1}\phi_{d(w)}^w\left(\mathcal{U}_{b_1 b_2\cdots b_{d(w)}}\right)}c_{j,m},\quad 1\leq j \leq d(w).
 \end{equation} 
Each node is recorded with its latest local objective value $f_{d(w)}(\phi_{d(w)}^w)$. If $w$ is a leaf node, i.e., $d(w)=K$, then $\phi_K^w$ yields a \emph{solution} to Problem~P2 as it determines the local objective for all $K$ agents and thus the global objective, which is the maximum of single-agent objective values recorded with nodes on the path from the root node to node $w$. Mathematically, the resultant global objective is 
\begin{equation}
    F(\phi_K^w)=\max_{j=1,\ldots,K} f_j(\phi_K^w).
\end{equation}
One can also verify that $\mathrm{dom}(\phi_{K}^w)=\{\mathcal{H}^j\}_{j=1}^{2^K}$, implying that $\phi_{K}^w$ specifies the mapped subset of all homogeneous voxel subsets. If $d(w)<K$, then $\phi(w)$ is a partial solution to Problem~P2, suggesting that node $w$ needs further subdivision for determining the next local objective value $f_{d(w)+1}$, of which the different possibilities constitute the set of child nodes of $w$. Each of these child nodes, say $\tilde{w}$, defines a subset-to-subset mapping with a finer granularity, $\phi_{d(w)+1}^{\tilde{w}}$, with domain $\mathrm{dom}(\phi_{d(w)+1}^{\tilde{w}})=\{\mathcal{U}_{b_1 b_2\cdots b_{d(w)+1}}\}_{b_i\in\{0,1\}}$. To keep the previous local objectives unchanged, we require $\phi_{d(w)+1}^{\tilde{w}}({\mathcal{U}_{b_1 b_2\cdots b_{d(w)}1}})\cup\phi_{d(w)+1}^{\tilde{w}}({\mathcal{U}_{b_1 b_2\cdots b_{d(w)}0}})=\phi_{d(w)}^{{w}}({\mathcal{U}_{b_1 b_2\cdots b_{d(w)}}})$ for all $b_1,\ldots,b_{d(w)}\in\{0,1\}$. In other words, constructing $\phi_{d(w)+1}^{\tilde{w}}$ is equivalent to selecting a subset $\phi_{d(w)+1}^{\tilde{w}}({\mathcal{U}_{b_1 b_2\cdots b_{d(w)}1}})\subset\phi_{d(w)}^{{w}}({\mathcal{U}_{b_1 b_2\cdots b_{d(w)}}})$ with size $|\mathcal{U}_{b_1 b_2\cdots b_{d(w)}1}|$ and assigning the de-selected ones as $\phi_{d(w)+1}^{\tilde{w}}({\mathcal{U}_{b_1 b_2\cdots b_{d(w)}0}})$ for all $b_1,\ldots,b_{d(w)}\in\{0,1\}$.

    Furthermore, we can incrementally rank the child nodes of $w$ in the order of ascending $f_{d(w)+1}({\phi}_{d(w)+1}^w)$ in an online manner, i.e., without listing and sorting all possible child nodes. This, as shown later, in most cases avoids enumerating all possible branches when combined with the depth-first search procedure and the derived pruning criteria. 
    To achieve this is equivalent to finding the $j$-best solution for the following subcarrier selection problem:
    \begin{equation*}\text{(P4($w$))} 
            \begin{aligned}
            \min\limits_{\{{\phi_{d(w)+1}(\mathcal{U}_{b_1 b_2\cdots b_{d(w)}1})}\}} \sum_{m\in\bigcup_{b_j}{\phi_{d(w)+1}(\mathcal{U}_{b_1 b_2\cdots b_{d(w)}1})}}c_{j,m} \\
            \mathrm{s. t.}\quad 
             |\phi_{d(w)+1}(\mathcal{U}_{b_1 b_2\cdots b_{d(w)}1})|=|\mathcal{U}_{b_1 b_2\cdots b_{d(w)}1}|, \\\forall b_1,\ldots,b_{d(w)} \in \{0,1\},\\
             \phi_{d(w)+1}(\mathcal{U}_{b_1 b_2\cdots b_{d(w)}1})\subset \phi_{d(w)}(\mathcal{U}_{b_1 b_2\cdots b_{d(w)}}),\\\forall b_1,\ldots,b_{d(w)} \in \{0,1\}.
            \end{aligned}
    \end{equation*}
    Since the selection of each $\phi_{d(w)+1}(\mathcal{U}_{b_1 b_2\cdots b_{d(w)}1})$ is decoupled with each other, this can be achieved by standard algorithms such as priority queues. 
    
    The example of a search tree with number of agents $K=3$ is illustrated in Fig. 2, where each node with depth $d$ is marked with its corresponding local objective value $f_{d}$.

\section{Optimal VoCa-PPA: Fast Tree-Search}    
    Given the compact search tree constructed in the preceding section, we present in this section two novel algorithms to accelerate the tree search via node pruning and agent ordering by exploiting the properties of VoCa-PPA.
    \subsection{Tree-Pruning Algorithm}
    The search tree constructed in the preceding section systematically organizes all possible solutions to Problem P2, represented by all its leaf nodes. However, in practice, it is computationally prohibitive to store all tree nodes and then perform an exhaustive search for the optimal solution. To address this issue, we hereby introduce an efficient tree search method combining DFS and problem-specific pruning criteria.

    \subsubsection{Depth-First Search with Priority}
    The DFS starts with visiting the root node $w_0$, and repeats visiting an unvisited child node of the last visited node, thereby increasing the search depth, until reaching a leaf node with the maximum depth. When a leaf node is visited, or the node visited has no unvisited child node, the algorithm backtracks to visit its parent node. In particular, when visiting a node with multiple child nodes, the one with the minimum local objective value is always prioritized. Not only is it a greedy heuristic which minimizes the cost for the current agent considered, but such a priority order can facilitate node pruning discussed in the sequel to reduce the number of nodes visited. Moreover, using the said method of incrementally ranking the nodes, the unvisited child nodes with lower priorities need not be explicitly defined and stored but are instantiated per request.
    \subsubsection{Stopping and Pruning Conditions}
    The optimal solution of the search tree minimizes the global objective among solutions associated with all of the tree's leaf nodes. The stopping and pruning conditions in the process of DFS build on the observation that every local objective constitutes a lower bound of the original objective, i.e., $F\geq f_{j}$ for all $j = 1,\ldots,K$. Thus, the objective value achieved by all descendants of node $w$ is lower bounded by the single-objective value achieved by node $w$ itself, i.e., $f_{d(w)}(\phi^w_{d(w)})$. 
    By applying this argument to the child nodes of the root node $w_0$, i.e., nodes with depth $1$, we argue that any depth $1$ node, say $w_j$, cannot develop into a better solution than $\phi_{\mathrm{best}}$ if $f_{1}(\phi^{w_j}_1)\geq F(\phi_{\mathrm{best}})$, and neither can any sibling nodes with a larger index than $w_j$ as all child nodes are ranked in ascending order of the local objective. This results in the global optimal condition stated as follows. 
    \begin{Proposition}[\emph{Stopping Condition}]
    \label{theorem-optimality-condition}
    \emph{During a DFS over the defined search tree in Section~VI-C, the current best solution to Problem~P2, denoted as $\phi_{\mathrm{best}}$, is optimal if 
    \begin{equation}\label{optimality_inequality}
        f_{1}(\phi_1^{w_j})\geq F(\phi_{\mathrm{best}}),
    \end{equation}
    where $w_j$ is the last visited depth-$1$ node.}
    \end{Proposition}
    \begin{proof}
        According to the sequence of node visiting in DFS, any unvisited leaf node, say node $w'$, must be a child node of either $w_j$ or a sibling node of $w_j$, say $w_{j'}$, with $j'>j$. In the former case, we have its solution value $F(\phi_K^{w'})\geq f_{1}(\phi_1^{w_j})\geq F(\phi_{\mathrm{best}})$. In the latter case, we have $F(\phi_K^{w'})\geq f_{1}(\phi_1^{w_{j'}})\geq f_{1}(\phi_1^{w_{j}})\geq F(\phi_{\mathrm{best}})$ due to the ascending order of local objectives. This completes the proof. 
    \end{proof}
    
    In the searching process, the updating of $\phi_{\mathrm{best}}$ is triggered if a newly found leaf node outperforms the current best solution, and the optimality condition \eqref{optimality_inequality} is checked if $\phi_{\mathrm{best}}$ is updated or a new depth $1$ node is visited. 

    The stopping condition for the global optimum is derived by bounding the global objective with the local ones achieved by child nodes of the root node $w_0$. On the other hand, each node, say node $w$, is associated with a \emph{sub-tree} with itself being the root node. Similar to the original tree, define an objective function $F_w(\phi_K^w)\triangleq \max_{d=d(w)+1\ldots, K} f_{d}(\phi_K^w)$ of the sub-tree for a mapping $\phi_K^w$ associated with a leaf node, which considers only a subset of agents instead of all $K$ agents. The optimal solution of the said sub-tree is defined to minimize $F_w(\phi_K^w)$. Thus, a natural question is: \emph{can we generalize Proposition~\ref{theorem-optimality-condition} to the sub-trees to enable further node pruning?} This is justified by the intuition that in the search for the global optimum, it suffices to look at the optimum of a sub-tree instead of all solutions of the sub-tree since a global optimum must also be a local optimum. This is formalized in the following lemma, with its proof omitted for brevity.
    \begin{Lemma}
    \emph{Let $\phi_K^w$ be a solution associated with a leaf node $w$ of a sub-tree. Then, the optimal solution of the sub-tree $\phi_K^{w^*}$ is at least at good as $\phi_K^w$ in terms of the global objective function $F$, i.e., $F(\phi_K^{w^*})\leq F(\phi_K^w)$.}
    \end{Lemma}
    Thus, the enumeration of the sub-tree's nodes can be stopped if its optimal solution is already found using a condition similar to \eqref{optimality_inequality}. The following proposition follows for pruning nodes which are unable to yield better solutions than visited nodes, with its proof omitted due to its similarity to that of Proposition~\ref{theorem-optimality-condition}.
    \begin{Proposition}\label{prop: pruning}(Pruning Criteria)
        \emph{During a DFS over the defined search tree in Section~VI-C, for a sub-tree associated with any node, say node $w$, its unvisited nodes can be pruned, i.e., need not be visited if 
        \begin{equation}
            f_{d(w)+1}(\phi_{d(w)+1}^{\tilde{w}_j})\geq F_w(\phi_{\mathrm{best},w}),
        \end{equation}
        where $\phi_{\mathrm{best},w}$ is the current best solution of the said sub-tree, and $\tilde{w}_j$ is the last visited child node of $w$.
        }
    \end{Proposition}
    A direct result follows: for a node $w$ with depth $K-1$, i.e., whose child nodes are leaf nodes of the search tree, upon visiting its leftmost child node, the remaining child node can be immediately pruned as the objective function of the sub-tree, $F_w$ is exactly $f_{K}$, giving $F_w(\phi_{\mathrm{best},w})=f_{K}(\phi_K^{w_1})$ which triggers Proposition~2. An example of tree search with pruning is illustrated in Fig.~\ref{fig_tree_search}, where nodes pruned according to Proposition~2 are marked with strides and Proposition~\ref{theorem-optimality-condition} is used for the optimality test.

    \subsection{Agent-Ordering Algorithm}\label{subsec: sensor-order}
    Determining the agent ordering can significantly affect the number of nodes visited before the algorithm finds the optimal solution. Selecting the agent order can be translated to determining the agent priority. The agents can be arranged in the descending order of their priorities. To this end, agent $1$ is given the highest priority because in the DFS process, the first child node we visit minimizes the objective $f_{1}$. Conditioned on $f_{1}$, we proceed to minimize the objective for agent $2$ with the second highest priority, and so forth. Following the intuition that the bottleneck agent should be given high priority, we propose an ordering heuristic based on a priority indicator, which is defined for each agent, say agent $i$, as 
    \begin{equation}
        \psi(i) = f^*_{i} =  \min_{|\mathcal{M}|=|\mathcal{V}_i|} \sum_{m\in \mathcal{M}} c_{i,m}.
    \end{equation}
    This indicator can be interpreted as the cost of the locally optimal mapping between non-sparse voxels and subcarriers for agent $i$ without considering other agents. A larger $\psi(i)$ indicates poorer channel states or more non-sparse voxels that need to be transmitted for agent $i$. From the tree-searching perspective, $\psi(i)$ is the objective lower bound obtained by visiting the very first child node of the root node if agent $i$ is visited first (see Proposition~\ref{theorem-optimality-condition}). As a result, letting agent $1$ be the one with the highest priority indicator yields the tightest initial lower bound. Thus we propose to arrange the agent index in descending order of the priority indicator, i.e., assigning index such that $\psi(i)\geq \psi(i')$ for any $1\leq i \leq i' \leq K$. 
    
    \subsection{Fast Tree-Search Algorithm}
    The fast tree search for optimal VoCa-PPA (i.e., solving Problem P2), which incorporates the two algorithms in the preceding subsections, is summarized in Algorithm~\ref{algo: optimal_pairing}.

    {
    \subsection{Complexity Analysis}
     The computation complexity of visiting each node by Algorithm~\ref{algo: optimal_pairing} in the defined tree can be divided into that of 1) local objective evaluation by \eqref{eqn: local_obj_rho_w}, which is $\mathcal{O}(M)$; 2) pruning/stopping determination by Proposition~\ref{theorem-optimality-condition} and Proposition~\ref{prop: pruning}, which is $\mathcal{O}(1)$; 3) enumeration of child nodes in ascending order of local objectives, which in the worst case $\mathcal{O}(N_\text{child}\log N_\text{child})$, where $N_\text{child}$ is the number of child nodes. Note that the last term is amortized by all $N_\text{child}$ nodes, and thus the amortized complexity per node is in fact upper bounded as $\mathcal{O}(\log N_\text{child})<\mathcal{O}(\log M!)=\mathcal{O}(M\log M)$. Meanwhile, the total number of visited nodes is upper bounded by $KN_\text{sol}$, where $N_\text{sol}$ is the number of solutions (leaf nodes) enumerated before the algorithm stops. Therefore, the worst case complexity of Algorithm~\ref{algo: optimal_pairing} is $\mathcal{O}(N_\text{sol}KM\log M)$. In the worst case, $N_\text{sol}$ can still reach the size of the full solution space, which is exponential in $M$. This is inevitable due to the NP-hardness of Problem P2. However, the empirical number of solutions visited is usually substantially lower than the worst case, thanks to the proposed fast tree-search algorithms. To illustrate the empirical complexity, $95$-th percentiles of $N_\text{sol}$ under different $M$ and $K=4$ are presented in Table~1 along with the solution space size, i.e., the number of all possible solutions. } 

\begin{table}[t]
{
\caption{Numbers of Enumerated Solutions and All Solutions}
\begin{tabular}{ccc}
\toprule
$M$  & 95-th percentile of $N_\text{sol}$ & Number of all solutions \\ \midrule
$8$  & $12$                               & $1.10\times 10^{4}$       \\
$16$ & $99$                               & $1.20\times 10^{11}$      \\
$32$ & $653$                              & $1.78\times 10^{27}$     \\ \bottomrule
\end{tabular}}
\end{table}
      
\begin{algorithm}[ht]
\caption{Fast Tree Search for Optimal VoCa-PPA}
\label{algo: optimal_pairing}
\textbf{Input:} Sparsity matrix $\mathbf{S}$ and channel matrix $\mathbf{H}$\;
\textbf{Prioritization:} Determine the agent indexing as elaborated in Section~\ref{subsec: sensor-order}\;
\textbf{Initialization:} Root node $w_0$ with $d(w_0)=0$\;
$\phi^*$ = DFS($w_0$)\;
Designate the optimal mapping $m^*(v)$ as an arbitrary one that satisfies $\phi^*$\;
Recover $\mathbf{A}^*$ from $m^*(v)$\;
$\,$\\
\textbf{function} DFS(w)\\
\ \ \ \ \textbf{for} node $\hat{w}$ \textbf{in} all non-root parent nodes of $w$ \textbf{do}\\
\ \ \ \ \ \ \ \ Invoke Proposition~2 to prune all unvisited nodes of $w$ if possible\;
\ \ \ \ \textbf{if} optimality test passes via Proposition~1 \textbf{then}\\
\ \ \ \ \ \ \ \ $\phi^*\leftarrow$ current best solution of the full tree\;
\ \ \ \ \ \ \ \ \textbf{return} optimal solution $\phi^*$\\
\ \ \ \ \textbf{if} $d(w)<K$ \textbf{then} \\
\ \ \ \ \ \ \ \ \textbf{while} $w$ has unvisited child nodes \textbf{do}\\
\ \ \ \ \ \ \ \ \ \ \ \ Create child node $\tilde{w}$ with $d(\tilde{w})=d(w)+1$\;
\ \ \ \ \ \ \ \ \ \ \ \ $\phi_{d(w)+1}^{\tilde{w}} \leftarrow$ the next best solution to  P4($w$)\;
\ \ \ \ \ \ \ \ \ \ \ \ Call DFS($\tilde{w}$)\;
\ \ \ \ \ \ \ \ \ \ \ \ \textbf{if} DFS($\tilde{w}$) returns optimal solution $\phi^*$ \textbf{then}\\
\ \ \ \ \ \ \ \ \ \ \ \ \ \ \ \ \textbf{return} optimal solution $\phi^*$\;
\ \ \ \ \textbf{return} continue search \\
\textbf{end function}
\end{algorithm}

\section{Experimental Results}
\label{sec: performance-evaluation}
\subsection{Experimental Settings}
We evaluate the performance of Spatial AirFusion on an ISEA system as illustrated in Fig.~\ref{fig_system}.  The channel between the fusion center and $K$ agents is assumed to follow i.i.d. Rician fading with the ratio between the power of line-of-sight (LoS) and non-LoS paths set as 3 dB and the path loss set as -15 dB. Following the Wi-Fi 6E standard, the total number of subcarriers in each resource block is $M=26$, each spanning a bandwidth of $B_\mathrm{sub}=120\ \text{kHz}$. The receive noise power per subcarrier is set as $-40$ dBm. Average pooling is adopted as the fusion function. 
The performance of Spatial AirFusion and baseline schemes is evaluated on the following two datasets.
\begin{itemize}
    \item \textbf{Synthetic dataset.} The synthetic dataset involves $K=4$ agents, each with a randomly generated feature map. The feature sparsity pattern $\bS$ is a random binary matrix with $1/3$ probability for each of its elements to be non-zero, while it is ensured that each column has at least one non-zero element, i.e., at least one agent has non-zero observations on each voxel. The simulated performance is averaged over $1000$ realizations with i.i.d. randomly generated channel matrices and sparsity patterns.
    \item \textbf{OPV2V dataset.} The OPV2V dataset \cite{xu2022opencood} considers a vehicle-to-vehicle communication scenario where an ego vehicle fuses sensory features from helping vehicles detect other vehicles in a traffic scene. A data frame involves two to five vehicles, one of which is selected as the ego vehicle. Each vehicle captures a LiDAR point cloud of the surrounding environment and objects, which is projected onto the ego vehicle's coordinates and processed by a PointPillar backbone into a two-dimensional local spatial feature map with $V_\mathsf{h}=256$ and $V_\mathsf{w}=352$ being the number of voxels along the height and width of the perception region, respectively. Each voxel is associated with a feature vector with dimension $L=128$, and thus the size of the local feature map is $128\times 256 \times 352$. {We find that in all voxels observed by all agents, over 90\% are empty, resulting in zero feature vectors, which conforms to the observations by \cite{Zhou2018CVPR, pointpillars}. }Therein, $50591$ out of $V_\mathsf{h}V_\mathsf{w} = 90112$ voxels are empty over all samples and all agents in the dataset, regarded as dummy, and waived of transmission for all evaluated methods. {The ego vehicle wirelessly aggregates the feature map from all other vehicles and inputs the fused feature map into an RPN, as in \cite{fcooper}, to obtain the vehicle detection result.  The detection performance is evaluated by comparing the downstream network output with the ground truth, measured by }the average precision (AP) at an Intersection over Union (IoU) threshold of 0.7. {It is defined as the area under the precision-recall curve resulting from the said detection model, where a detected bounding box is considered true-positive if it overlaps with a ground-truth bounding box with an IoU higher than $0.7$\cite{Zhou2018CVPR}.}
\end{itemize}
\begin{figure*}[t]
\centering
\subfigure[]{\includegraphics[height=4.6cm]{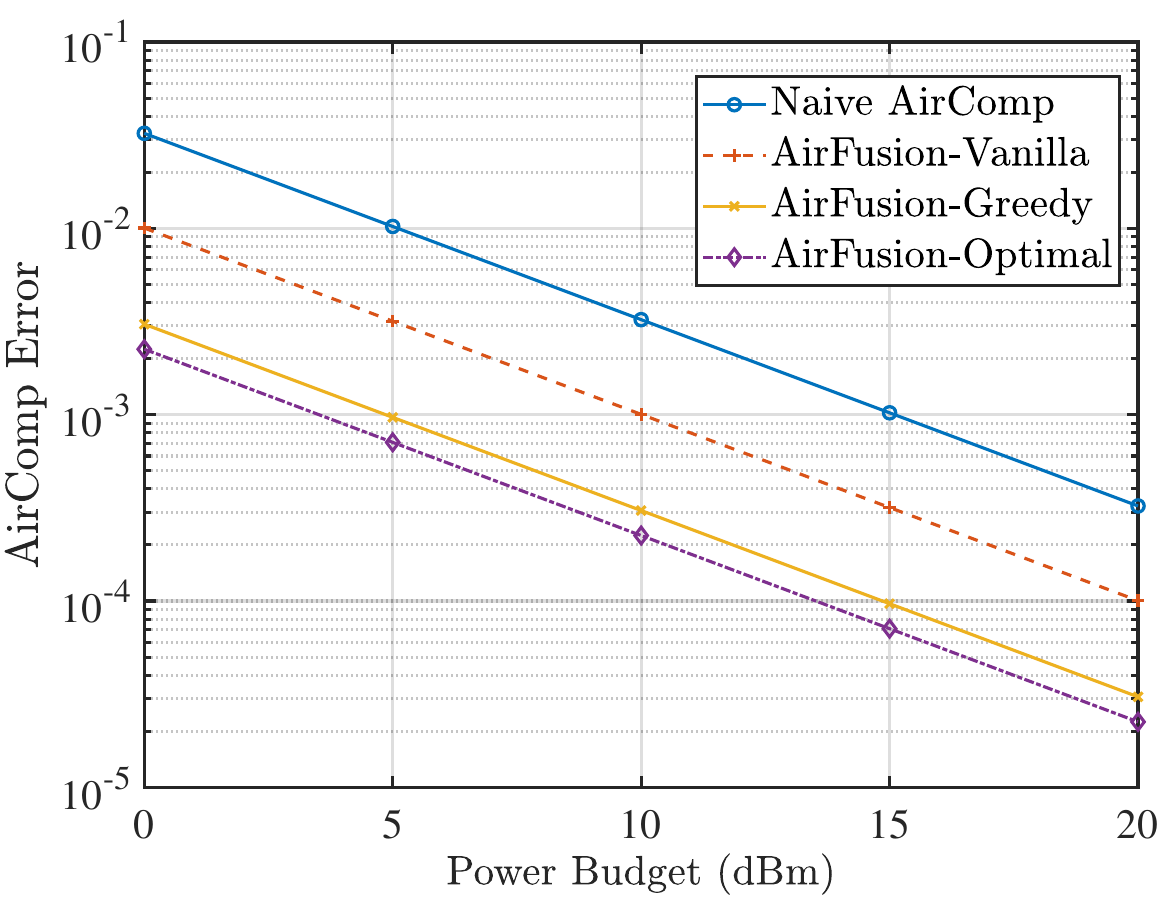}\label{fig_syn_tx_power}}
\hspace{2.8cm}
\subfigure[]{\includegraphics[height=4.6cm]{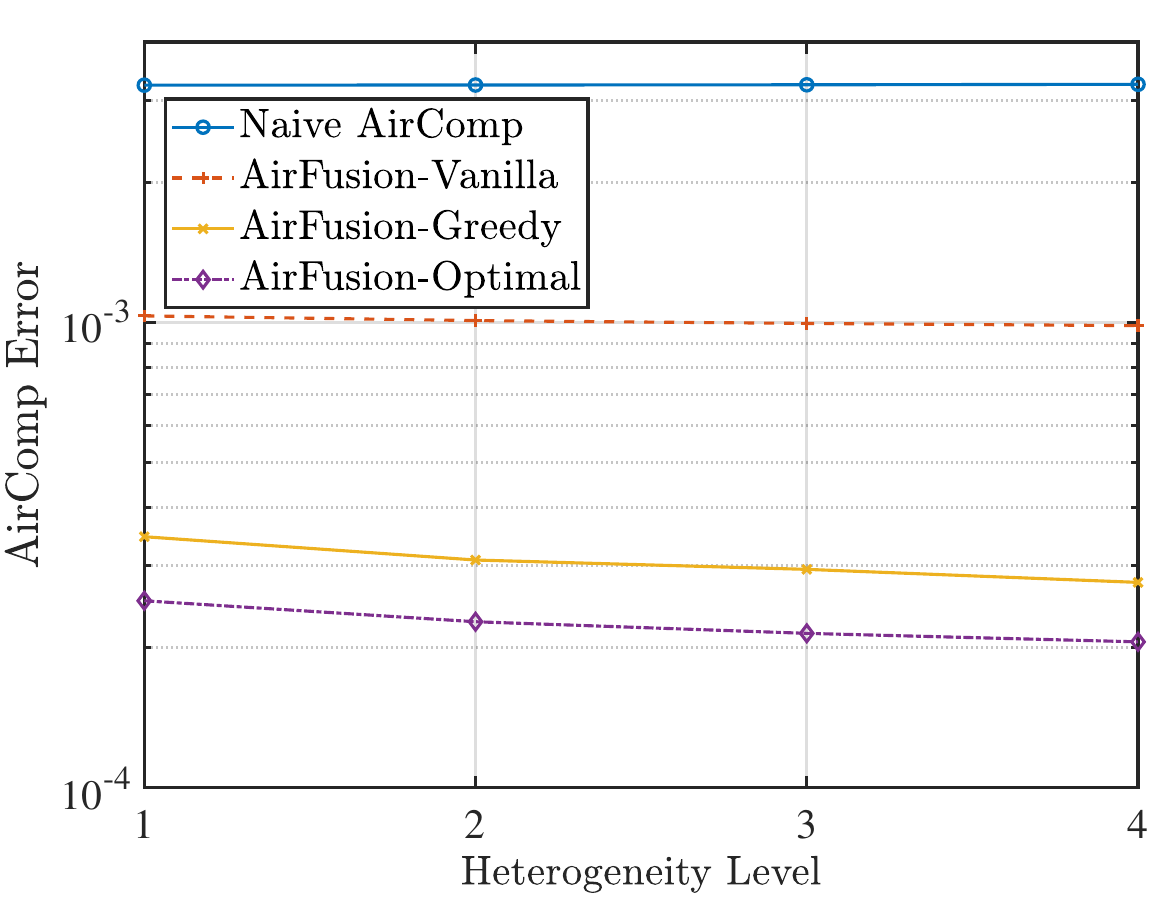}\label{fig_syn_hetero_level}}
\vspace{-0.2cm}
\caption{The performance of variants of Spatial AirFusion and naive AirComp on the synthetic dataset. }
\label{fig_syn}
\end{figure*}
We compare the performance of Spatial AirFusion controlled by VoCa-PPA with three benchmarking schemes, called \emph{naive AirComp}, \emph{digital air interface}, and \emph{AirFusion-Vanilla}. 
\begin{itemize}
    \item \textbf{Naive AirComp.} Naive AirComp aggregates each voxel over the air on an assigned subcarrier similar to Spatial AirFusion, but does not involve the feedback of the feature sparsity matrix. Thus, all agents  participate in  AirComp over all subcarriers regardless of sparsity\cite{GX2020TWC}. The subcarriers are allocated in sequential order and the receive SNR, which is fixed for all subcarriers in each coherence block, is chosen such that all agents' power constraints are satisfied. 
    \item \textbf{Digital air interface.} The scheme corresponds to the conventional digital broadband orthogonal-access approach, {where each agent is assigned a subset of subcarriers for feature uploading.} On the agent side, each feature coefficient {is encoded into $2$ to $5$ bits, depending on the desired latency-precision tradeoff, via uniform quantization. The radio resource management scheme with max-marginal-rate subcarrier assignment and equal power allocation, proposed and shown to be near-optimal in\cite{1437359}, is adopted. Then the communication latency is calculated using Shannon capacity given the assigned subcarrier and power.}  After receiving data from all agents, the server decodes the bits stream to reconstruct features. 
    \item {\textbf{AirFusion-Vanilla.} This scheme implements the system architecture and operations of AirFusion as in Section~II and Section~III, but pairs voxels with subcarriers in a sequential order without optimization. Given the default pairing, the power is optimally allocated using Lemma~1.}
\end{itemize}
\begin{figure*}[t]
\centering
\subfigure[ ]{\includegraphics[height=4.6cm]{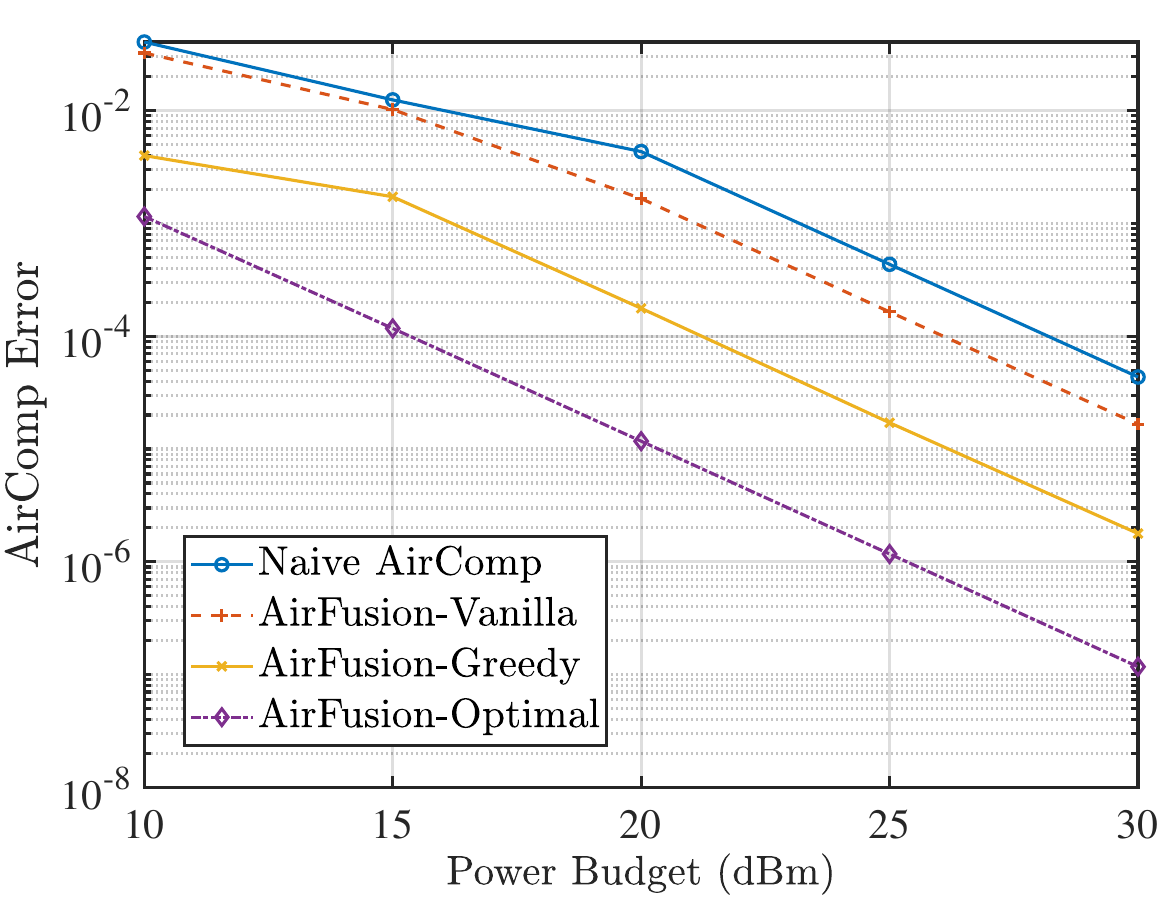}\label{fig_opencood_3_err}}
\hspace{2.9cm}
\subfigure[ ]{\includegraphics[height=4.6cm]{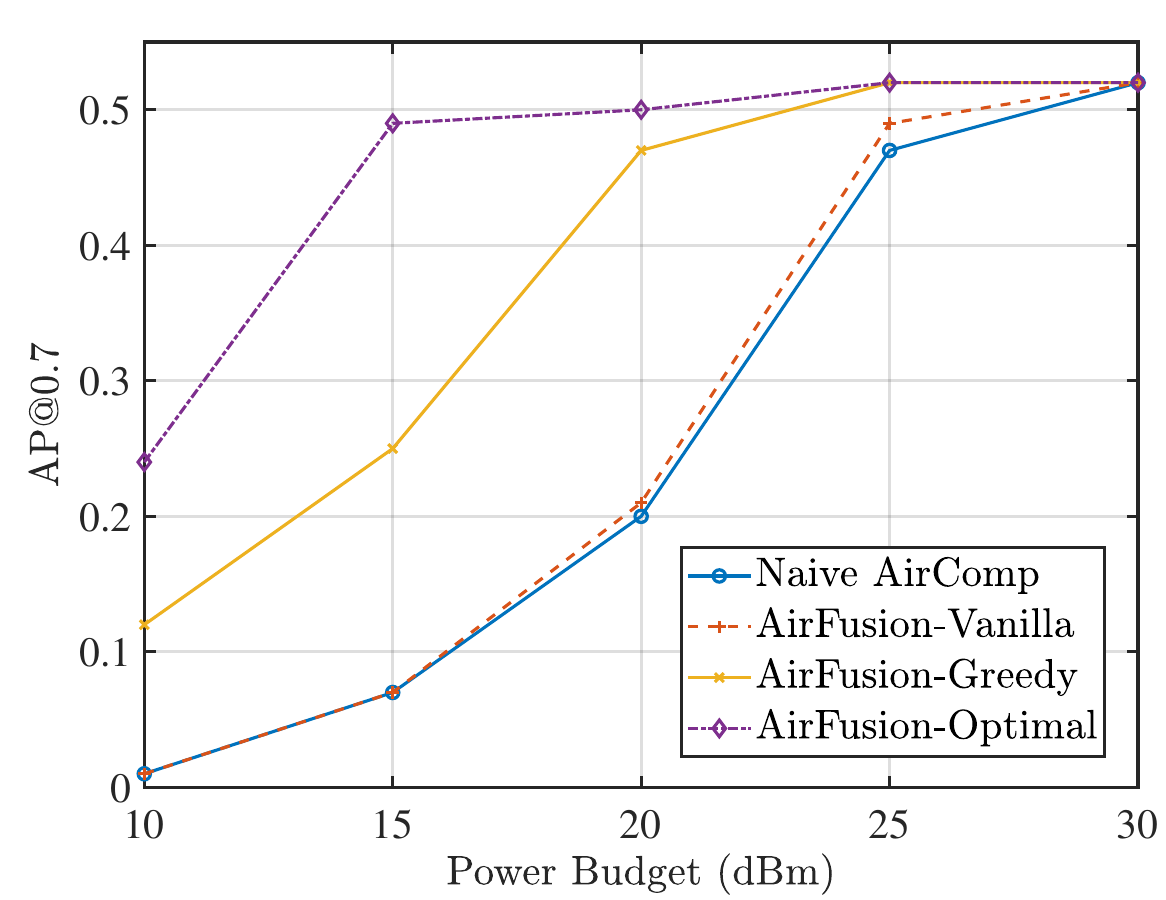}\label{fig_opencood_3_acc}}
\vspace{-0.2cm}
\caption{The performance of variants of Spatial AirFusion and naive AirComp on the OPV2V dataset with number of CAVs $K=3$.}
\label{fig_opencood_3}
\end{figure*}

\begin{figure*}[t]
\centering
\subfigure[ ]{\includegraphics[height=4.6cm]{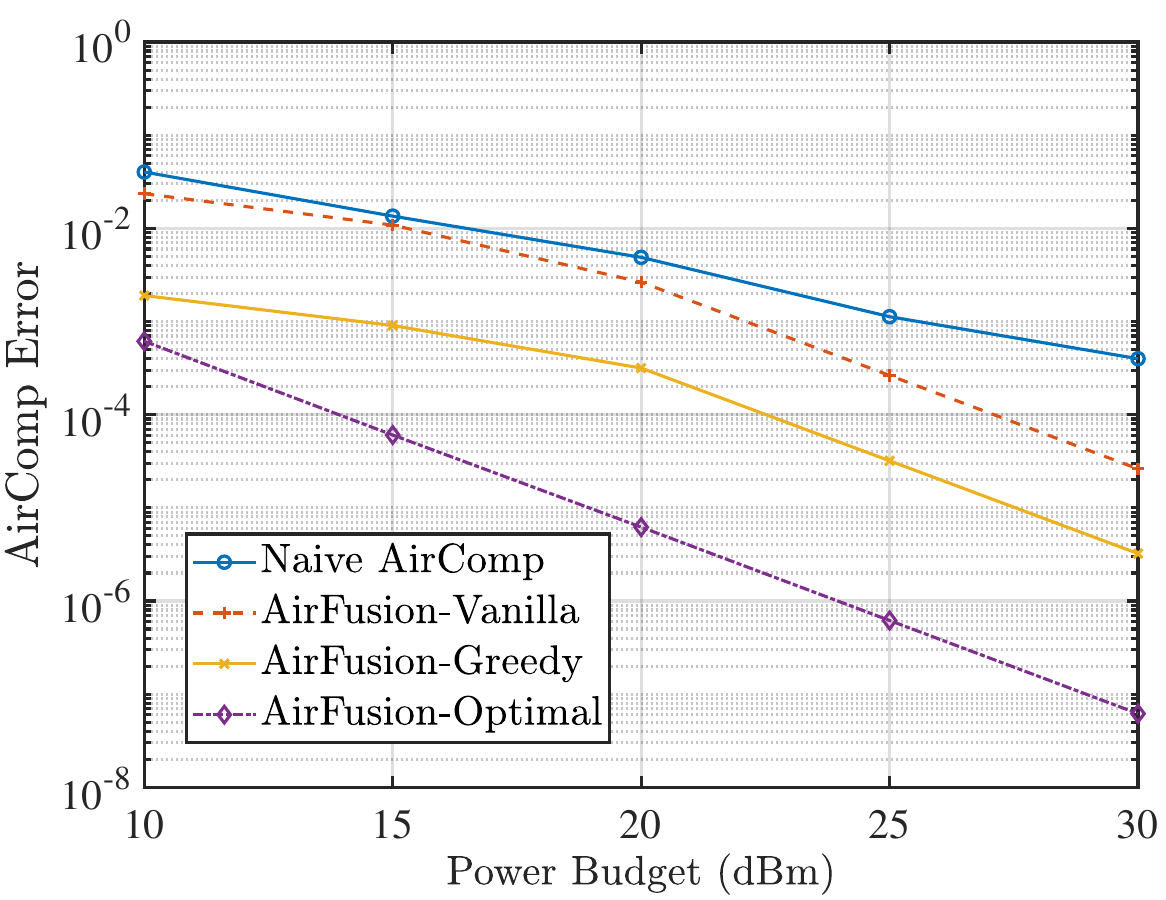}\label{fig_opencood_4_err}}
\hspace{2.9cm}
\subfigure[ ]{\includegraphics[height=4.6cm]{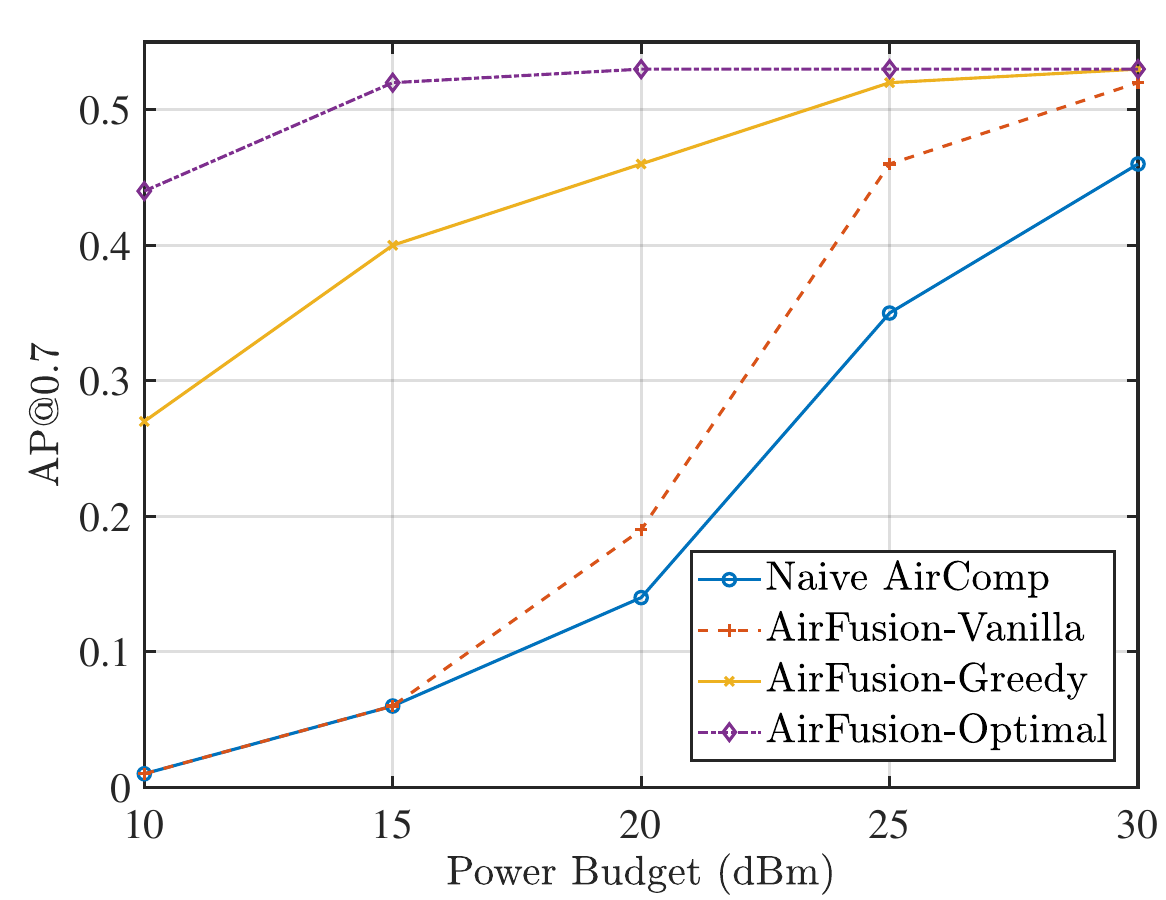}\label{fig_opencood_4_acc}}
\vspace{-0.2cm}
\caption{The performance of variants of Spatial AirFusion and naive AirComp on the OPV2V dataset with number of CAVs $K=4$.}
\label{fig_opencood_4}
\end{figure*}
\subsection{Performance Evaluation on Synthetic Datasets}

First, the performance of Spatial AirFusion and naive AirComp is evaluated on the synthetic dataset. {We test Spatial AirFusion controlled by Greedy VoCa-PPA in Algorithm~\ref{algo: greedy_control} and Optimal VoCa-PPA in Algorithm~\ref{algo: optimal_pairing}, termed ``AirFusion-Greedy'' and ``AirFusion-Optimal'', respectively. The performance is measured by AirComp error, defined as the mean square error of feature aggregation results compared with the ideal ground-truth case, i.e., \eqref{eqn: fusion-ground-truth}}. The curves of AirComp error versus transmit power budget on each agent are plotted in Fig.~\ref{fig_syn_tx_power}. We observe that the sparsity-aware Spatial AirFusion protocol design can roughly reduce the AirComp error by 70\% with AirFusion-Vanilla which does not optimize subcarrier allocation. This can be attributed to the reduction in communication overhead combined with smarter power allocation by exploiting sparsity of spatial features. On top of vanilla Spatial AirFusion, incorporating optimal VoCa pairing further improves the Spatial AirFusion performance as observed from the greedy and optimal cases. The small optimality gap between the algorithms renders greedy VoCa-PPA a close-to-optimal heuristic with low computational complexity.

{Fixing the transmit power budget at 10 dBm, we vary the sparsity heterogeneity measured by the entropy of the empirical distribution of homogeneous subsets, as given by $-\sum_{q=1}^{2^K}\frac{|\mathcal{H}^q|}{V}\log\frac{|\mathcal{H}^q|}{V}$. It reaches the maximum when voxels are uniformly distributed to all homogenous subsets and zero when all voxels belong to the same homogeneous subset.} The AirComp error performance against heterogeneity level is plotted in Fig.~\ref{fig_syn_hetero_level}. We find a reduction in AirComp error when the heterogeneity level increases for Spatial AirFusion but not for naive AirComp that does not exploit spatial sparsity. The reason is that with a more heterogeneous voxel distribution, the proposed framework is provisioned with more degrees-of-freedom for VoCa pairing. {This aligns with the intuition that in the extreme case where all voxels belong to the same homogeneous subset, the gain of the proposed approach diminishes since the homogeneity of voxels renders arbitrary VoCa allocation optimal.}

\subsection{Performance Evaluation on the OPV2V dataset}

The experimental results of Spatial AirFusion and naive AirComp obtained on the OPV2V dataset are presented in Fig.~\ref{fig_opencood_3}. The curves of AirComp error versus power budget for 3 and 4 participating vehicles, as plotted in Figs.~\ref{fig_opencood_3_err}~and~\ref{fig_opencood_4_err}, respectively, show a trend similar to that on the synthetic dataset where Spatial AirFusion significantly outperforms naive AirComp. In terms of inference accuracy shown in Figs.~\ref{fig_opencood_3_acc}~and~\ref{fig_opencood_4_acc}, which is measured by the average precision at an intersection over union (IoU) threshold of 0.7, Spatial AirFusion delivers substantially better performance than naive AirComp. As the transmit power budget reaches 20 dBm, the accuracy of Spatial AirFusion saturates at about 50\%, which is due to the inherent robustness of the perception model that tolerates a certain amount of distortion in the aggregated features without losing accuracy.
    \begin{figure}[t]
    \centering
\includegraphics[height=4.6cm]{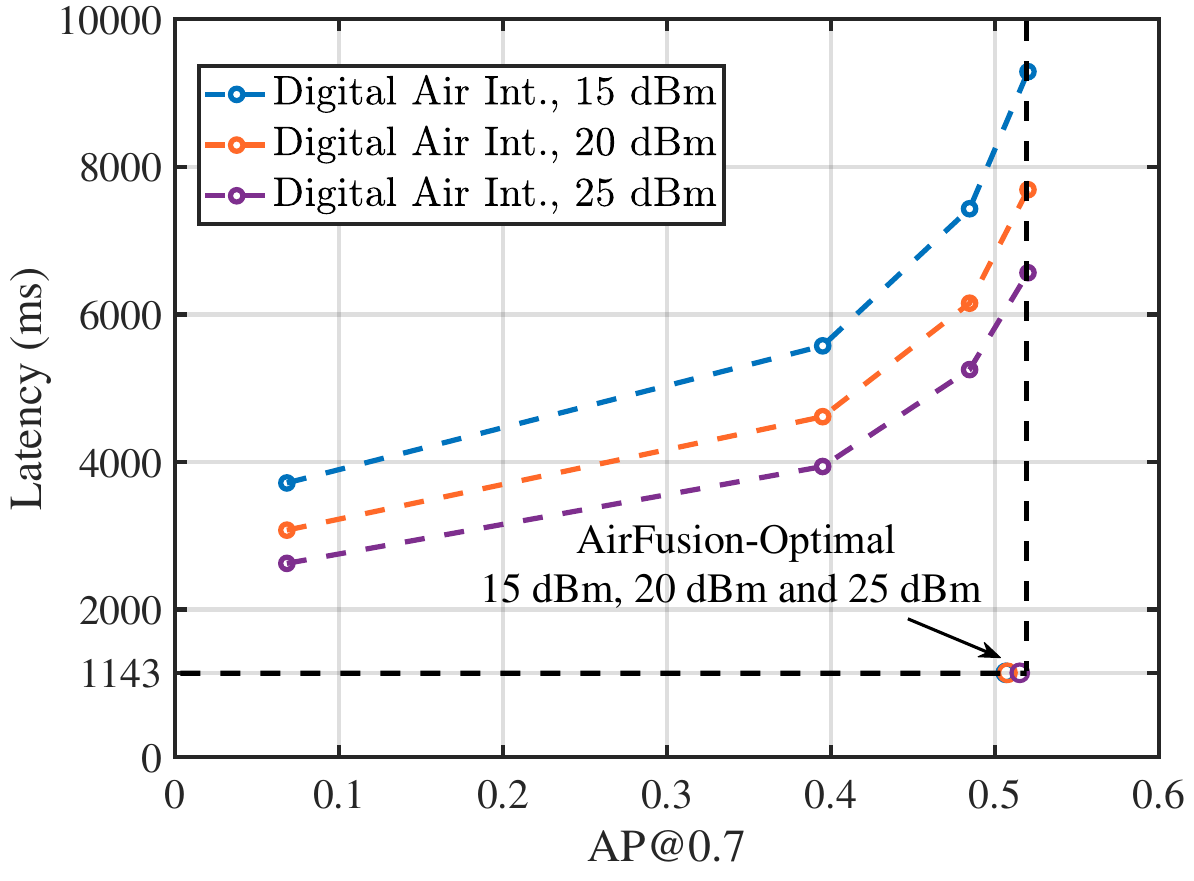}
    \vspace{-0.2cm}
    \caption{The tradeoff between communication latency (in \emph{millisecond} (ms)) and sensing performance measured in AP@0.7 for Spatial AirFusion and digital air interface different transmit power budgets on the OPV2V dataset. }
    \label{fig_dai_AirFusion}
\end{figure}

Finally, we compare the Pareto fronts of latency-precision tradeoff for digital air interface and Spatial AirFusion. The communication latency is defined as the average transmission time required to aggregate all features of a single {perception instance}. For AirComp, the said latency is independent of the transmit power and given by $L_\mathrm{H}=L N_\mathrm{v}/B_\mathrm{sub}$, where $\tilde{N}_\mathrm{v}$ is the average number of non-sparse voxels in each LiDAR frame. For digital air interface, the latency $L_\mathrm{D}$ depends on the total number of OFDM rounds required to transmit all features. Therein, a lower transmit power budget or poorer channels lead to lower communication rates and thus more required rounds. Given a certain transmit power budget, the latency-precision tradeoff in digital air interface is regulated by feature quantization resolution varying from 2 bits to 5 bits. The results are plotted in Fig.~\ref{fig_dai_AirFusion}. We observe that under the same precision requirement and transmit power budget, Spatial AirFusion can reduce the latency by up to an order of magnitude. For example, digital air interface requires $5$-bit quantization to achieve a target precision of $51\%$ at $25$ dBm power budget, where the resultant latency is $6,565$ ms. In contrast, Spatial AirFusion completes transmission in only $1,143$ ms, reducing the latency by $5.74$ times. Two factors contribute to the latency reduction. The first is the exploitation of waveform superposition to avoid orthogonal transmission of each agent's feature. Second, through the sparsity pattern feedback, a substantial number of sparse voxels need not be transmitted in the case of Spatial AirFusion. 
\vspace{-0.2cm}
\section{Concluding Remarks}
\label{sec: conclusions}
In this paper, we have presented the framework of Spatial AirFusion, a broadband task-oriented air interface targeting multi-agent environment perception tasks. The Spatial AirFusion protocol is developed to exploit spatial feature sparsity, a critical property of perception models, for enhancing communication efficiency. A mixed-integer programming problem, i.e., the VoCa-PPA problem, is formulated for joint allocation of power and subcarriers to maximize the minimum received SNR among all voxels. We solve this problem by designing a low-complexity greedy VoCa pairing algorithm and also an optimal tree search approach via exploiting useful properties of the problem structure. Experimental results show significant improvement in error suppression, sensing performance, and latency reduction compared with conventional approaches.

{ We acknowledge that several assumptions have been made in this paper to simplify the exposition, which motivate further studies. Identical feature variance is assumed across all agents, the relaxation of which requires feature statistics-aware power control. Symbol-level synchronization and sufficient channel coherence time across all agents are assumed to facilitate AirComp design. However, under high mobility, such conditions may not hold, requiring advanced scheduling and air-interface designs.}

This work opens up several research directions on task-oriented communication schemes for ISEA. For example, digital Spatial AirFusion can be developed for better compatibility with existing digital systems, enabling incorporation of digital transmission techniques such as modulation and coding schemes. Another interesting topic is the interplay between Spatial AirFusion and more sophisticated physical layer techniques such as MIMO. { In existence of strong resource heterogeneity across agents, asynchronous feature aggregation could be necessitated, and the relevant scheduling and fusion schemes warrant future studies.} In addition, integrating Spatial AirFusion with semantic data sourcing, which broadcasts low-dimensional queries to trigger transmission on semantically relevant agents, can further reduce communication cost\cite{SEMDAS}. 
\vspace{-0.2cm}
\bibliographystyle{IEEEtran}
\bibliography{Edge-Inference}

\begin{thebibliography}{10}
\providecommand{\url}[1]{#1}
\csname url@samestyle\endcsname
\providecommand{\newblock}{\relax}
\providecommand{\bibinfo}[2]{#2}
\providecommand{\BIBentrySTDinterwordspacing}{\spaceskip=0pt\relax}
\providecommand{\BIBentryALTinterwordstretchfactor}{4}
\providecommand{\BIBentryALTinterwordspacing}{\spaceskip=\fontdimen2\font plus
\BIBentryALTinterwordstretchfactor\fontdimen3\font minus \fontdimen4\font\relax}
\providecommand{\BIBforeignlanguage}[2]{{%
\expandafter\ifx\csname l@#1\endcsname\relax
\typeout{** WARNING: IEEEtran.bst: No hyphenation pattern has been}%
\typeout{** loaded for the language `#1'. Using the pattern for}%
\typeout{** the default language instead.}%
\else
\language=\csname l@#1\endcsname
\fi
#2}}
\providecommand{\BIBdecl}{\relax}
\BIBdecl

\bibitem{ITUR2023}
ITU-R, ``Framework and overall objectives of the future development of {I}{M}{T} for 2030 and beyond,'' [Online] https://www.itu.int/rec/R-REC-M.2160-0-202311-I/en, 2023.

\bibitem{Chen2023arxiv}
X.~Chen, K.~B. Letaief, and K.~Huang, ``On the view-and-channel aggregation gain in integrated sensing and edge {A}{I},'' \emph{IEEE J. Sel. Areas Commun.}, vol.~42, no.~9, pp. 2292--2305, Sep. 2024.

\bibitem{Huang2023TWC}
Z.~Liu, Q.~Lan, A.~E. Kalør, P.~Popovski, and K.~Huang, ``Over-the-air multi-view pooling for distributed sensing,'' \emph{IEEE Trans. Wireless Commun.}, vol.~23, no.~7, pp. 7652--7667, Jul. 2024.

\bibitem{Dingzhu24TWC}
D.~Wen, P.~Liu, G.~Zhu, Y.~Shi, J.~Xu, Y.~C. Eldar, and S.~Cui, ``Task-oriented sensing, computation, and communication integration for multi-device edge {A}{I},'' \emph{IEEE Trans. Wireless Commun.}, vol.~23, no.~3, pp. 2486--2502, Mar. 2024.

\bibitem{Wu2023Network}
H.~Xing, G.~Zhu, D.~Liu, H.~Wen, K.~Huang, and K.~Wu, ``Task-oriented integrated sensing, computation and communication for wireless edge {A}{I},'' \emph{IEEE Netw.}, vol.~37, no.~4, pp. 135--144, July/August 2023.

\bibitem{Shi2022JSAC}
K.~B. Letaief, Y.~Shi, J.~Lu, and J.~Lu, ``Edge artificial intelligence for 6{G}: Vision, enabling technologies, and applications,'' \emph{IEEE J. Sel. Areas Commun.}, vol.~40, no.~1, pp. 5--36, Jan. 2022.

\bibitem{GX2021WCM}
G.~Zhu, J.~Xu, K.~Huang, and S.~Cui, ``Over-the-air computing for wireless data aggregation in massive {I}o{T},'' \emph{IEEE Wireless Commun.}, vol.~28, no.~4, pp. 57--65, Aug. 2021.

\bibitem{GX2020TWC}
G.~Zhu, Y.~Wang, and K.~Huang, ``Broadband analog aggregation for low-latency federated edge learning,'' \emph{IEEE Trans. Wireless Commun.}, vol.~19, no.~1, pp. 491--506, Jan. 2020.

\bibitem{Deniz2020TWC}
M.~M. Amiri and D.~G\"und\"uz, ``Federated learning over wireless fading channels,'' \emph{IEEE Trans. Wireless Commun.}, vol.~19, no.~5, pp. 3546--3557, May 2020.

\bibitem{CMZ2021JSAC}
M.~Chen, D.~Gündüz, K.~Huang, W.~Saad, M.~Bennis, A.~V. Feljan, and H.~V. Poor, ``Distributed learning in wireless networks: Recent progress and future challenges,'' \emph{IEEE J. Sel. Areas Commun.}, vol.~39, no.~12, pp. 3579--3605, Dec. 2021.

\bibitem{Zhou2018CVPR}
Y.~Zhou and O.~Tuzel, ``Voxel{N}et: End-to-end learning for point cloud based 3{D} object detection,'' in \emph{Proc. IEEE/CVF Conf. Comput. Vision Pattern Recogn. (CVPR)}, Salt Lake City, UT, USA, Jun. 18--23, 2018.

\bibitem{Rukhovich2022WACV}
D.~Rukhovich, A.~Vorontsova, and A.~Konushin, ``Im{V}oxel{N}et: Image to voxels projection for monocular and multi-view general-purpose 3d object detection,'' in \emph{Proc. IEEE/CVF Winter Conf. Appl. Computer Vision (WACV)}, Waikoloa, HI, USA, Jan. 3--8, 2022.

\bibitem{Xie2022arxiv}
E.~Xie, Z.~Yu, D.~Zhou, J.~Philion, A.~Anandkumar, S.~Fidler, P.~Luo, and J.~M. Alvarez, ``M${^2}${B}{E}{V}: Multi-camera joint 3{D} detection and segmentation with unified birds-eye view representation,'' [Online] https://arxiv.org/pdf/2204.05088.pdf, 2022.

\bibitem{Lan2023TWC}
Q.~Lan, Q.~Zeng, P.~Popovski, D.~Gündüz, and K.~Huang, ``Progressive feature transmission for split classification at the wireless edge,'' \emph{IEEE Trans. Wireless Commun.}, vol.~22, no.~6, pp. 3837--3852, Jun. 2023.

\bibitem{Shao2023TWC}
J.~Shao, Y.~Mao, and J.~Zhang, ``Task-oriented communication for multidevice cooperative edge inference,'' \emph{IEEE Trans. Wireless Commun.}, vol.~22, no.~1, pp. 73--87, Jan. 2023.

\bibitem{Zhang2020CM}
J.~Shao and J.~Zhang, ``Communication-computation trade-off in resource-constrained edge inference,'' \emph{IEEE Commun. Mag.}, vol.~58, no.~12, pp. 20--26, Dec. 2020.

\bibitem{Zhou2020IoTJ}
X.~Huang and S.~Zhou, ``Dynamic compression ratio selection for edge inference systems with hard deadlines,'' \emph{IEEE Internet Things J.}, vol.~7, no.~9, pp. 8800--8810, Sep. 2020.

\bibitem{Deniz2022ISIT}
S.~F. Yilmaz, B.~Hasırcıoğlu, and D.~Gündüz, ``Over-the-air ensemble inference with model privacy,'' in \emph{Proc. IEEE Int. Symp. Inf. Theory (ISIT)}, Espoo, Finland, Jun. 26 -- Jul. 1, 2022.

\bibitem{Wen2023TWC}
Z.~Zhuang, D.~Wen, Y.~Shi, G.~Zhu, S.~Wu, and D.~Niyato, ``Integrated sensing-communication-computation for over-the-air edge {A}{I} inference,'' \emph{IEEE Trans. Wireless Commun.}, vol.~23, no.~4, pp. 3205--3220, Apr. 2024.

\bibitem{Gastpar2011TIT}
B.~Nazer and M.~Gastpar, ``Compute-and-forward: Harnessing interference through structured codes,'' \emph{IEEE Trans. Inf. Theory}, vol.~57, no.~10, pp. 6463--6486, Oct. 2011.

\bibitem{Katabe2016arxiv}
O.~Abari, H.~Rahul, and D.~Katabi, ``Over-the-air function computation in sensor networks,'' [Online] https://arxiv.org/pdf/1612.02307.pdf, 2016.

\bibitem{Xiaowen2020TWC}
X.~Cao, G.~Zhu, J.~Xu, and K.~Huang, ``Optimized power control for over-the-air computation in fading channels,'' \emph{IEEE Trans. Wireless Commun.}, vol.~19, no.~11, pp. 7498--7513, Nov. 2020.

\bibitem{GX2019IOTJ}
G.~Zhu and K.~Huang, ``M{I}{M}{O} over-the-air computation for high-mobility multimodal sensing,'' \emph{IEEE Internet Things J.}, vol.~6, no.~4, pp. 6089--6103, Aug. 2019.

\bibitem{Xiaowen2021TWC}
X.~Cao, G.~Zhu, J.~Xu, and K.~Huang, ``Cooperative interference management for over-the-air computation networks,'' \emph{IEEE Trans. Wireless Commun.}, vol.~20, no.~4, pp. 2634--2651, Apr. 2021.

\bibitem{Bennis2021Globecom}
M.~Krouka, A.~Elgabli, C.~ben Issaid, and M.~Bennis, ``Communication-efficient split learning based on analog communication and over the air aggregation,'' in \emph{Proc. IEEE Global Commun. Conf. (GLOBECOM)}, Madrid, Spain, Dec. 7-11, 2021.

\bibitem{Liu2021WCL}
T.~Qin, W.~Liu, B.~Vucetic, and Y.~Li, ``Over-the-air computation via broadband channels,'' \emph{IEEE Wireless Commun. Lett.}, vol.~10, no.~10, pp. 2150--2154, Oct. 2021.

\bibitem{Zhang2023TWC}
Z.~Lin, H.~Liu, and Y.-J.~A. Zhang, ``C{F}{L}{I}{T}: Coexisting federated learning and information transfer,'' \emph{IEEE Trans. Wireless Commun.}, vol.~22, no.~11, pp. 8436--8453, Nov. 2023.

\bibitem{Tao2021TWC}
N.~Zhang and M.~Tao, ``Gradient statistics aware power control for over-the-air federated learning,'' \emph{IEEE Trans. Wireless Commun.}, vol.~20, no.~8, pp. 5115--5128, Aug. 2021.

\bibitem{ZJun2024TWC}
Y.~Sun, Z.~Lin, Y.~Mao, S.~Jin, and J.~Zhang, ``Channel and gradient-importance aware device scheduling for over-the-air federated learning,'' \emph{IEEE Trans. Wireless Commun.}, vol.~23, no.~7, pp. 6905--6920, Jul. 2024.

\bibitem{Eldar2021TSP}
T.~Sery, N.~Shlezinger, K.~Cohen, and Y.~C. Eldar, ``Over-the-air federated learning from heterogeneous data,'' \emph{IEEE Trans. Signal Process.}, vol.~69, pp. 3796--3811, Jun. 2021.

\bibitem{pointpillars}
A.~H. Lang, S.~Vora, H.~Caesar, L.~Zhou, J.~Yang, and O.~Beijbom, ``Point{P}illars: {F}ast encoders for object detection from point clouds,'' in \emph{Proc. IEEE Conf. Comput. Vision Pattern Recogn. (CVPR)}, Long Beach, CA, USA, Jun. 16-20, 2019.

\bibitem{Vucetic2020TWC}
W.~Liu, X.~Zang, Y.~Li, and B.~Vucetic, ``Over-the-air computation systems: Optimization, analysis and scaling laws,'' \emph{IEEE Trans. Wireless Commun.}, vol.~19, no.~8, pp. 5488--5502, Aug. 2020.

\bibitem{ehrgott2006discussion}
M.~Ehrgott, ``A discussion of scalarization techniques for multiple objective integer programming,'' \emph{Ann. Oper. Res.}, vol. 147, no.~1, pp. 343--360, 2006.

\bibitem{xu2022opencood}
R.~Xu, H.~Xiang, X.~Xia, X.~Han, J.~Li, and J.~Ma, ``O{P}{V}2{V}: {A}n open benchmark dataset and fusion pipeline for perception with vehicle-to-vehicle communication,'' in \emph{Proc. IEEE Int. Conf. Robot. Autom. (ICRA)}, Philadelphia, PA, USA, May 23--27, 2022.

\bibitem{coordinate_offset}
S.-W. Kim, Z.~J. Chong, B.~Qin, X.~Shen, Z.~Cheng, W.~Liu, and M.~H. Ang, ``Cooperative perception for autonomous vehicle control on the road: Motivation and experimental results,'' in \emph{Proc. IEEE/RSJ Int. Conf. Intell. Robots Syst. (IROS)}, Tokyo, Japan, Oct. 1--5, 2013.

\bibitem{avr}
H.~Qiu, F.~Ahmad, F.~Bai, M.~Gruteser, and R.~Govindan, ``A{V}{R}: {A}ugmented {V}ehicular {R}eality,'' in \emph{Proc. 16th Annu. Int. Conf. Mobile Syst., Appl., Services (MobiSys)}, Munich, Germany, Jun. 10--15, 2018.

\bibitem{fusion_survey}
C.~Xiang, C.~Feng, X.~Xie, B.~Shi, H.~Lu, Y.~Lv, M.~Yang, and Z.~Niu, ``Multi-sensor fusion and cooperative perception for autonomous driving: A review,'' \emph{IEEE Intell. Transp. Syst. Mag.}, vol.~15, no.~5, pp. 36--58, Aug. 2023.

\bibitem{nomographic_approx}
S.~Limmer, J.~Mohammadi, and S.~Stańczak, ``A simple algorithm for approximation by nomographic functions,'' in \emph{Proc. 53rd Annu. Allert. Conf. Commun. Control Comput. (Allerton)}, Monticello, IL, USA, Sep. 29 -- Oct. 2, 2015.

\bibitem{channel_est_ofdm}
Y.~Liu, Z.~Tan, H.~Hu, L.~J. Cimini, and G.~Y. Li, ``Channel estimation for {O}{F}{D}{M},'' \emph{IEEE Commun. Surveys Tuts.}, vol.~16, no.~4, pp. 1891--1908, Fourthquarter 2014.

\bibitem{limited_feedback}
J.~Chen, R.~A. Berry, and M.~L. Honig, ``Limited feedback schemes for downlink {O}{F}{D}{M}{A} based on sub-channel groups,'' \emph{IEEE J. Sel. Areas Commun.}, vol.~26, no.~8, pp. 1451--1461, Oct. 2008.

\bibitem{gu2023twc}
Y.~Gu, C.~She, Z.~Quan, C.~Qiu, and X.~Xu, ``Graph neural networks for distributed power allocation in wireless networks: {A}ggregation over-the-air,'' \emph{IEEE Trans. Wireless Commun.}, vol.~22, no.~11, pp. 7551--7564, Mar. 2023.

\bibitem{Shi2020TWC}
K.~Yang, T.~Jiang, Y.~Shi, and Z.~Ding, ``Federated learning via over-the-air computation,'' \emph{IEEE Trans. Wireless Commun.}, vol.~19, no.~3, pp. 2022--2035, Mar. 2020.

\bibitem{rpn}
S.~Ren, K.~He, R.~Girshick, and J.~Sun, ``Faster {R}-{C}{N}{N}: {T}owards real-time object detection with region proposal networks,'' \emph{IEEE Trans. Pattern Anal. Mach. Intell.}, vol.~39, no.~6, pp. 1137--1149, Jun. 2017.

\bibitem{xu_dpca}
X.~Chen, E.~G. Larsson, and K.~Huang, ``Analog {M}{I}{M}{O} communication for one-shot distributed principal component analysis,'' \emph{IEEE Trans. Signal Process.}, vol.~70, pp. 3328--3342, Jun. 2022.

\bibitem{wen2024twc}
D.~Wen, X.~Jiao, P.~Liu, G.~Zhu, Y.~Shi, and K.~Huang, ``Task-oriented over-the-air computation for multi-device edge {A}{I},'' \emph{IEEE Trans. Wireless Commun.}, vol.~23, no.~3, pp. 2039--2053, Jul. 2024.

\bibitem{fcooper}
Q.~Chen, X.~Ma, S.~Tang, J.~Guo, Q.~Yang, and S.~Fu, ``F-{C}ooper: Feature based cooperative perception for autonomous vehicle edge computing system using 3{D} point clouds,'' in \emph{Proc. ACM/IEEE Symp. Edge Comput.}, Washington, DC, USA, Nov. 7--9, 2019.

\bibitem{v2xsim}
Y.~Li, D.~Ma, Z.~An, Z.~Wang, Y.~Zhong, S.~Chen, and C.~Feng, ``V2{X}-{S}im: {M}ulti-agent collaborative perception dataset and benchmark for autonomous driving,'' \emph{IEEE Robot. Autom. Lett.}, vol.~7, no.~4, pp. 10\,914--10\,921, Jul. 2022.

\bibitem{hardware_testbed}
Q.~Zhang, Z.~Feng, and P.~Zhang, ``Hardware testbed design and performance evaluation for {I}{S}{A}{C} enabled {C}{A}{V}s,'' in \emph{Integrated Sensing and Communications}, F.~Liu, C.~Masouros, and Y.~C. Eldar, Eds.\hskip 1em plus 0.5em minus 0.4em\relax Singapore: Springer Singapore, 2023, pp. 567--586.

\bibitem{BELHOUL201497}
L.~Belhoul, L.~Galand, and D.~Vanderpooten, ``An efficient procedure for finding best compromise solutions to the multi-objective assignment problem,'' \emph{Comput. Oper. Res.}, vol.~49, pp. 97--106, Sep. 2014.

\bibitem{1437359}
K.~Kim, Y.~Han, and S.-L. Kim, ``Joint subcarrier and power allocation in uplink {O}{F}{D}{M}{A} systems,'' \emph{IEEE Commun. Lett.}, vol.~9, no.~6, pp. 526--528, Jun. 2005.

\bibitem{SEMDAS}
K.~Huang, Q.~Lan, Z.~Liu, and L.~Yang, ``Semantic data sourcing for 6{G} edge intelligence,'' \emph{IEEE Commun. Mag.}, vol.~61, no.~12, pp. 70--76, Dec. 2023.

\end{thebibliography}
\begin{IEEEbiography}
[{\includegraphics[width=1in,height=1.25in,clip,keepaspectratio]{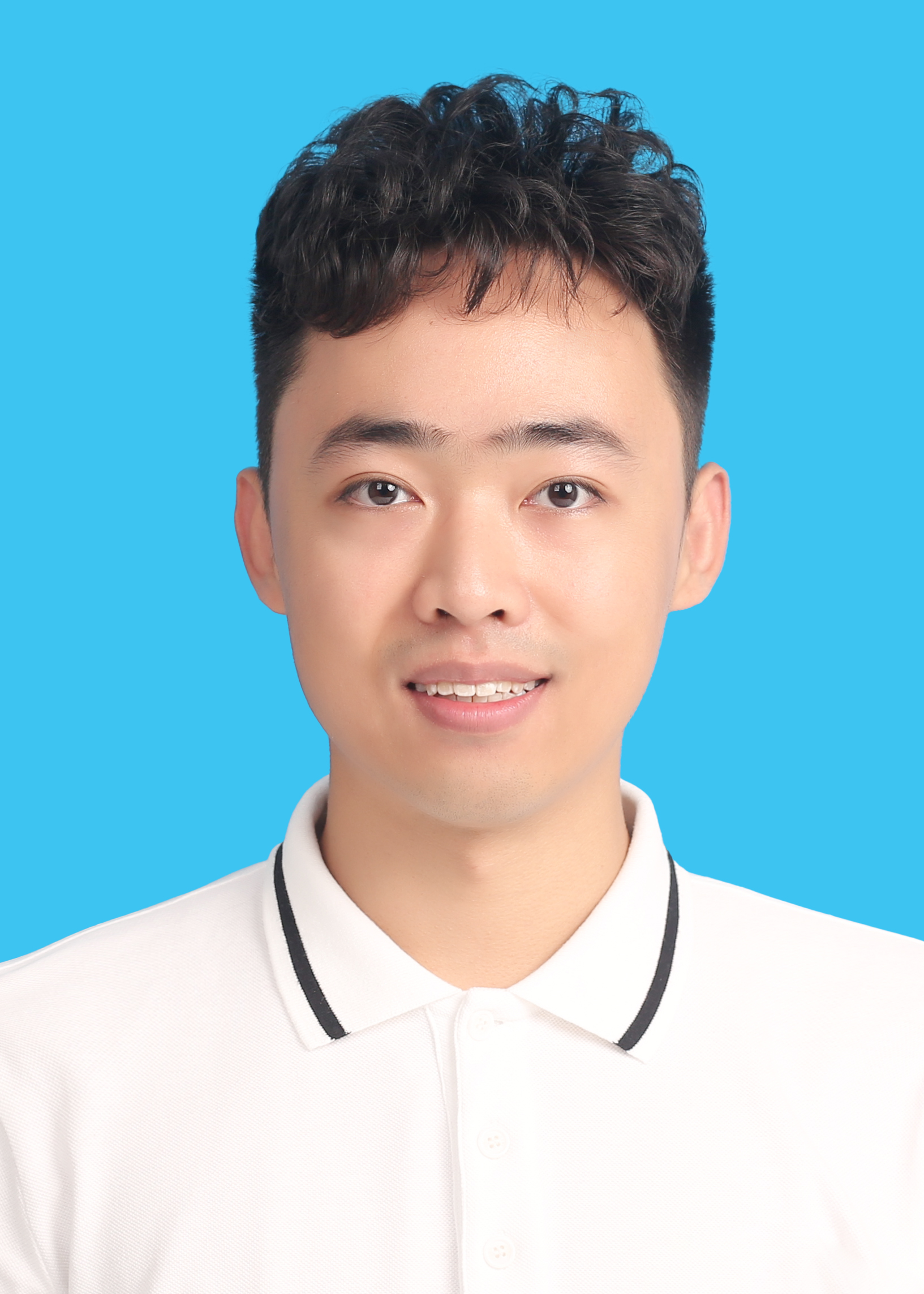}}]{Zhiyan Liu} (Graduate Student Member, IEEE) received the B.Eng. degree from the Dept. of Electronic Engineering, Tsinghua University, Beijing, in 2021. He is currently working towards the Ph.D. degree with Dept. of Electrical and Electronic Engineering, The University of Hong Kong (HKU), Hong Kong. His recent research interests include edge intelligence and distributed sensing in 6G wireless networks. He was a recipient of Hong Kong Ph.D. Fellowship.
\end{IEEEbiography}

\begin{IEEEbiography}
[{\includegraphics[width=1in,height=1.25in,clip,keepaspectratio]{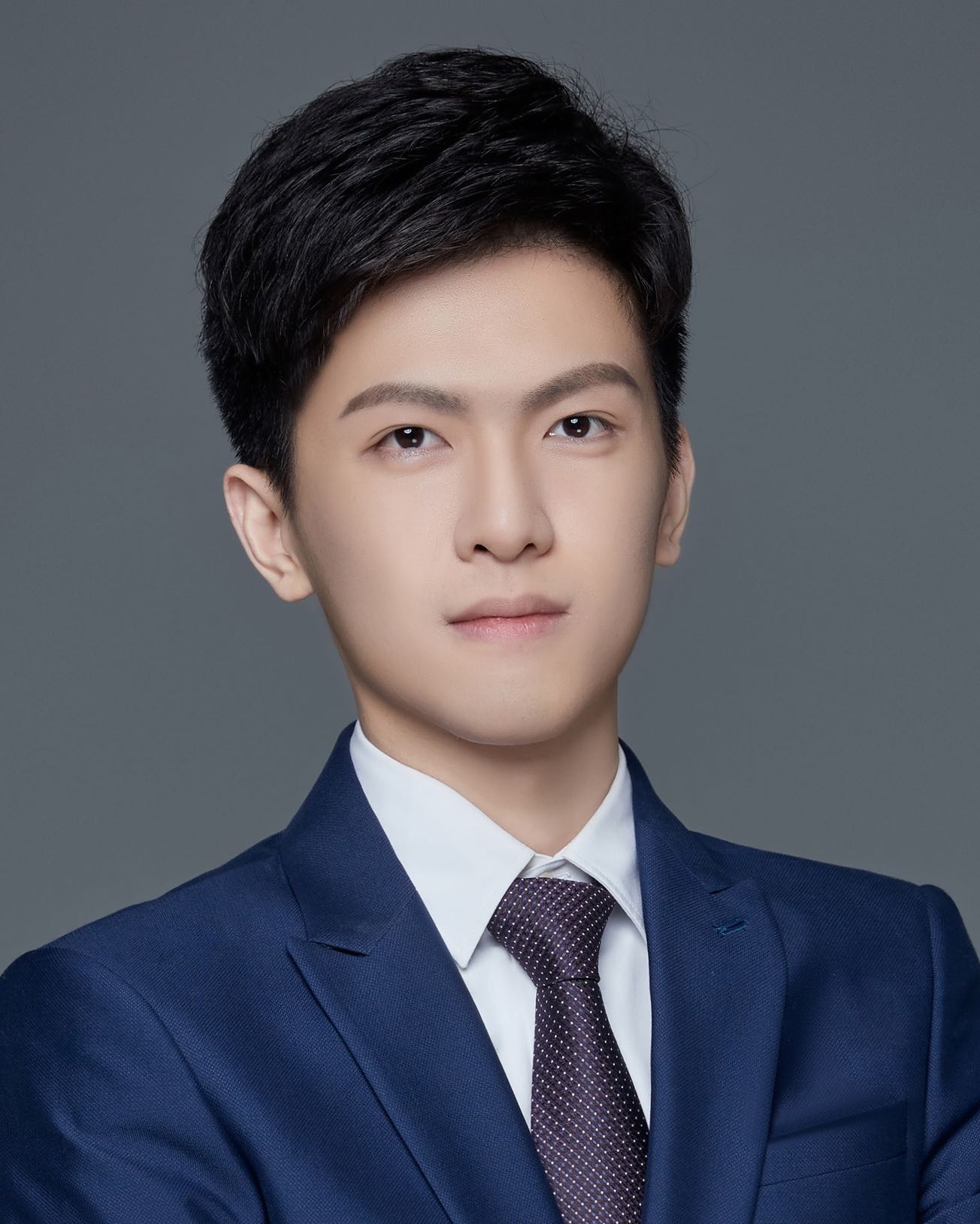}}]{Qiao Lan} (Member, IEEE) received the B.Eng. degree (with honors) from the Southern University of Science and Technology, Shenzhen, in 2019, and the Ph.D. degree from The University of Hong Kong, Hong Kong, in 2023. He is now a senior engineer with a research laboratory in the wireless communications industry. His recent research interests include AI algorithms and systems in wireless networks. 
\end{IEEEbiography}

\begin{IEEEbiography}
[{\includegraphics[width=1in,height=1.25in,clip,keepaspectratio]{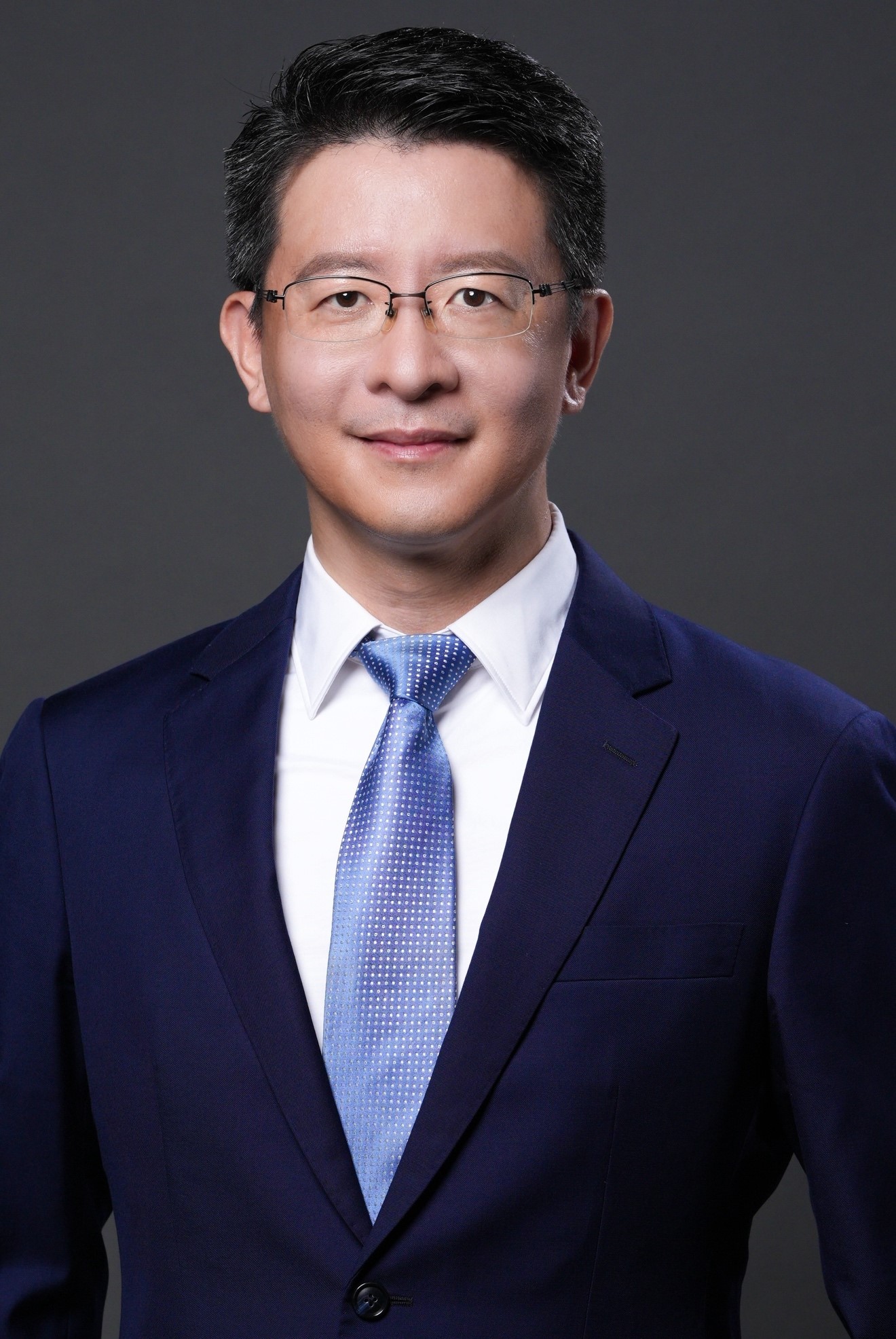}}]{Kaibin Huang} (Fellow, IEEE) received the B.Eng. and M.Eng. degrees from the National University of Singapore and the Ph.D. degree from The University of Texas at Austin, all in electrical engineering. He is a Professor and the Head at the Dept. of Electrical and Electronic Engineering, The University of Hong Kong (HKU), Hong Kong. He received the IEEE Communication Society’s 2021 Best Survey Paper, 2019 Best Tutorial Paper, 2019 and 2023 Asia–Pacific Outstanding Paper, 2015 Asia–Pacific Best Paper Award, and the best paper awards at IEEE GLOBECOM 2006 and IEEE/CIC ICCC 2018. He has been named as a Highly Cited Researcher by Clarivate in 2019-2023 and an AI 2000 Most Influential Scholar (Top 30 in Internet of Things) in 2023-2024. He was an IEEE Distinguished Lecturer of both the IEEE Communications Society and the IEEE Vehicular Technology Society. He is a member of the Engineering Panel of Hong Kong Research Grants Council (RGC) and a RGC Research Fellow (2021 Class). He received the Outstanding Teaching Award from Yonsei University, South Korea, in 2011. He is an Area Editor of IEEE Transactions on Wireless Communications, IEEE Transactions on Machine Learning in Communications and Networking, and IEEE Transactions on Green Communications and Networking. Previously, he served on the Editorial Boards for IEEE Journal on Selected Areas in Communications (JSAC) and IEEE Wireless Communication Letters. He has guest edited special issues of IEEE JSAC, IEEE Journal of Selected Areas in Signal Processing, and IEEE Communications Magazine, and IEEE Network. He served as the Lead Chair for the Wireless Communications Symposium of IEEE Globecom 2017 and the Communication Theory Symposium of IEEE GLOBECOM 2023 and 2014, and the TPC Co-chair for IEEE PIMRC 2017 and IEEE CTW 2023 and 2013. He is the founding President of the HKU chapter of National Academy of Inventors. 

\end{IEEEbiography}

\end{document}